\pdfoutput=1
\documentclass[acmsmall,screen]{acmart}

\makeatletter

\newcommand\addauthornote[1]{%
  \if@ACM@anonymous\else
    \g@addto@macro\addresses{\g@addto@macro\@currentauthors{{\small *}}}
  \fi}

\newcommand\@addauthornotemark[1]{\let\@tmpcnta\0
    \g@addto@macro\@currentauthors{\footnotemark\relax\let\0\@tmpcnta}}

\makeatother

\acmJournal{PACMPL}
\acmVolume{1}
\acmNumber{CONF} % CONF = POPL or ICFP or OOPSLA
\acmArticle{1}
\acmYear{2018}
\acmMonth{1}
\acmDOI{} % \acmDOI{10.1145/nnnnnnn.nnnnnnn}
\startPage{1}

\setcopyright{none}

\usepackage{booktabs}
\usepackage{subcaption}

\usepackage{amsmath}
\usepackage{stmaryrd}
\usepackage{thm-restate}
\usepackage{ifthen}

\newcommand{\iflabeldefined}[3]{%
  \ifcsname r@#1\endcsname
    #2% If the label is defined
  \else
    #3% If the label is undefined
  \fi
}

\newboolean{showappendices}
\setboolean{showappendices}{true}

\usepackage{environ}
\usepackage{cleveref}
\usepackage{algorithm}
\usepackage[noend]{algpseudocode}
\usepackage[color]{changebar}
\usepackage{xspace}
\usepackage{bussproofs}
\usepackage{stackengine}
\usepackage{wrapfig}
\usepackage{multirow}
\usepackage{pgfplots}
\pgfplotsset{width=10cm,compat=1.9}

\DeclareFontFamily{U}{mathb}{\hyphenchar\font45}
\DeclareFontShape{U}{mathb}{m}{n}{
      <5> <6> <7> <8> <9> <10> gen * mathb
      <10.95> mathb10 <12> <14.4> <17.28> <20.74> <24.88> mathb12
      }{}
\DeclareSymbolFont{mathb}{U}{mathb}{m}{n}
\DeclareFontSubstitution{U}{mathb}{m}{n}
\DeclareMathSymbol{\sqsubsetneq}    {3}{mathb}{"88}
\DeclareMathSymbol{\sqsupsetneq}    {3}{mathb}{"89}
\DeclareMathSymbol{\varsqsubsetneq} {3}{mathb}{"8A}
\DeclareMathSymbol{\varsqsupsetneq} {3}{mathb}{"8B}

\newcommand{\rone}{(\emph{i})~}
\newcommand{\rtwo}{(\emph{ii})~}
\newcommand{\rthree}{(\emph{iii})~}
\newcommand{\rfour}{(\emph{iv})~}

\newcommand{\grammar}{G}
\newcommand{\examples}{\mathcal{E}}

\newcommand{\proofspacing}{0.5em}

\newcommand{\dens}[1]{\den{#1}}
\newcommand{\den}[1]{\llbracket #1 \rrbracket}
\newcommand{\dennon}[2]{\llbracket #2 \rrbracket_{#1}}
\newcommand{\denanon}[2]{\llbracket #2 \rrbracket_{#1}^\sharp}
\newcommand{\aden}[1]{\llbracket #1 \rrbracket^\sharp}
\newcommand{\dena}[1]{{\den{#1}}^\sharp}

\newcommand{\totalbenchmarks}{430}

\newcommand{\LL}{\mathcal L}
\newcommand{\lowerb}{\textsc{l}}
\newcommand{\upper}{\textsc{u}}

\newcommand{\semgus}{SemGuS\xspace}
\newcommand{\sygus}{SyGuS\xspace}

\newcommand{\toolname}{\textsc{Moito}\xspace}
\newcommand{\baseline}{\textsc{baseline}\xspace}
\newcommand{\virtualbest}{\textsc{best}\xspace}
\newcommand{\toolnamegfa}{\textsc{Moito-n}\xspace}
\newcommand{\toolnamenogfa}{\textsc{Moito-h}\xspace}

\newcommand{\code}[1]{\texttt{#1}}

\newcommand{\rc}[1]{\texttt{#1}} % regex character

\newcommand{\mz}{\textbf{0}}
\newcommand{\mo}{\textbf{1}}
\newcommand{\mi}{\textbf{I}}
\newcommand{\hole}{\square}

\newcommand{\monup}{\uparrow}
\newcommand{\mondown}{\downarrow}
\newcommand{\monvec}{m}

\newcommand{\AlphaId}{{\textit{Alpha}}}
\newcommand{\NumId}{{\textit{Num}}}
\newcommand{\RowId}{{\textit{Row}}}
\newcommand{\IntId}{{\textit{Int}}}

\definecolor{forest}{RGB}{83,147,49}
\newcommand{\changed}[1]{{{#1}}}
\definecolor{bottlegreen}{rgb}{0.0,0.42,0.31}

\begin{document}

\title{Automating Pruning in Top-Down Enumeration for Program Synthesis Problems with Monotonic Semantics}

\author{Keith J.C. Johnson}
\authornote{Both authors contributed equally to the paper.}
\orcid{0000-0002-3766-5204}             %% \orcid is optional
\affiliation{
  \institution{University of Wisconsin-Madison}
  \country{USA}                    %% \country is recommended
}
\email{keithj@cs.wisc.edu}

\author{Rahul Krishnan}
\addauthornote{1}
\orcid{0000-0003-0230-5185}             %% \orcid is optional
\affiliation{
  \institution{University of Wisconsin-Madison}
  \country{USA}                    %% \country is recommended
}
\email{rahulk@cs.wisc.edu}

\author{Thomas Reps}
\orcid{0000-0002-5676-9949}             %% \orcid is optional
\affiliation{
  \institution{University of Wisconsin-Madison}
  \country{USA}                    %% \country is recommended
}
\email{reps@cs.wisc.edu}

\author{Loris D'Antoni}
\orcid{0000-0001-9625-4037}             %% \orcid is optional
\affiliation{
  \institution{University of California, San Diego}
  \country{USA}                    %% \country is recommended
}
\email{ldantoni@ucsd.edu}

\begin{abstract}
% %
In top-down enumeration for program synthesis, abstraction-based pruning uses an abstract domain to approximate the set of possible values that a partial program, when completed, can output on a given input.
If the set does not contain the desired output, the partial program and all its possible completions can be pruned.
In its general form, abstraction-based pruning requires manually designed, domain-specific abstract domains and semantics, and thus has only been used in domain-specific synthesizers.

This paper provides sufficient conditions under which a form of abstraction-based pruning can be automated for arbitrary synthesis problems in the general-purpose Semantics-Guided Synthesis (\semgus) framework without requiring manually-defined abstract domains.
We show that if the semantics of the language for which we are synthesizing programs exhibits some monotonicity properties, one can obtain an abstract interval-based semantics for free from the concrete semantics of the programming language, and use such semantics to effectively prune the search space.
We also identify a condition that ensures
% demonstrate that
such abstract semantics can be used to compute a precise abstraction of the set of values that a program derivable from a given hole in a partial program can produce.
These precise abstractions make abstraction-based pruning more effective.

We implement our approach in a tool, \toolname, which can tackle synthesis problems defined in the \semgus framework.
\toolname can automate interval-based pruning without any a-priori knowledge of the problem domain, and solve synthesis problems that previously required domain-specific, abstraction-based synthesizers---e.g., synthesis of regular expressions, CSV file schema, and imperative programs from examples.
\end{abstract}

\maketitle

\section{Introduction}
\label{sec:introduction}

The goal of program synthesis is to find a program that satisfies a given specification, typically given as a logical formula or a set of input-output examples.
A simple synthesis technique is \textit{top-down enumeration}, where one enumerates programs derivable from a regular-tree grammar---i.e., the search space of the synthesis problem---by iteratively expanding \textit{partial programs}---i.e., syntax trees that contain one or more \textit{hole} symbols to be filled by as-yet undetermined sub-trees---using productions in the grammar.
% 
% One of the limitations of top-down enumeration, when compared to its bottom-up counterpart, is that it has to enumerate both complete and partial programs, thus resulting in some extra programs being considered~\cite{mybook}.
%
Top-down enumeration is effective in practice when one can aggressively prune the search space of enumerated programs.
Pruning is done by discarding partial programs that cannot be completed to yield a program that satisfies the given specification.
% 
% When such programs are identified, they can be discarded/pruned, together with their infinitely many possible completions.

\subsubsection*{Domain-Specific Pruning Strategies}
% Designing pruning strategies for top-down enumeration has often been treated as an art form where skilled engineers identify clever domain-specific insights for proving when partial programs will not succeed.
% 
Clever pruning strategies have resulted in top-down enumeration approaches that can quickly synthesize imperative programs \cite{so2017synthesizing}, regular expressions \cite{lee2016alpharegex}, SQL queries \cite{wang2017scythe}, and Datalog programs \cite{si2018datalog}.
However, coming up with the domain-specific insights to allow enumeration to scale is a challenging task, to the extent that such insights are the main contribution of many papers.

Consider the problem of synthesizing a regular expression that accepts the (positive) example string $11$.
During enumeration, a partial regular expression $0\cdot \hole_R$ can be pruned because no matter what regular expression we replace $\hole_R$ with, the string $11$ will never be accepted.
This simple technique was the main contribution that allowed the tool AlphaRegex \cite{lee2016alpharegex} to synthesize regular expressions of non-trivial size.

Similarly, consider the problem of synthesizing an imperative program that is consistent with a set of input-output examples, and  suppose that we have the partial program \texttt{x\,:=\,0; y\,:=\,$\hole_E$}, where $\hole_E$ can be replaced with an expression over the variables $x,y$ and integer constants.
This program can be discarded if the synthesis goal requires mapping the input $x=5$ to output $x=7$: any instantiation of the assignment to $y$ still results in the input $x=5$ being mapped to $x=0$.

Despite the time and care that went into developing these techniques, they cannot be automatically tailored to potentially similar tasks outside of their rigid predefined domains.

\subsubsection*{Automating Pruning}
In this paper, we show that the aforementioned pruning strategies can be ``rationally reconstructed'' as instances of a common \textbf{automated domain-agnostic} framework for pruning partial programs.
Given a partial program with holes $P$, our framework uses abstract interpretation to abstract the semantics of all the possible program completions on a specific input, and obtain, in the form of an interval $[l,u]$, a superset of the possible values that a completed program might compute.
When the value we want the synthesized program to compute is not in the interval $[l,u]$, the  program $P$ can be pruned.
For example, the set of possible regular expressions that $0 \cdot \hole_R$ can take on can be overapproximated by the interval $[\emptyset, 0\cdot (0\mid 1)^*]$, where a regular expression $r$ is inside this interval if and only if $\emptyset \subseteq L(r) \subseteq L(0 \cdot (0\mid 1)^*)$.
However, the positive example $11$ is not accepted by any regular expression in this interval, because $11 \notin L(0 \cdot (0\mid 1)^*)$, meaning we can prune out $0 \cdot \hole_R$.
Similarly, for the imperative-program pruning example, the output value of $x$ produced by \texttt{x\,:=\,0; y\,:=\,$\hole_E$} on input $x=5$ can be expressed as (and hence ``overapproximated by'') the interval $[0,0]$.
Because the desired output $7$ is not in this interval, we can prune this partial program.

Our work also identifies sufficient conditions on the semantics of the programs appearing in the search space that allow one to generate the operations needed to compute such abstractions \textbf{automatically}.
If the semantics of the programs in the search space satisfies monotonicity conditions (that often can be automatically checked), our framework provides for free a \emph{precise} interval-based abstraction.
The key insight is that given a monotonically increasing function $f$ and an interval of possible inputs $[l,u]$, the tightest interval that encloses the result of evaluating $f$ on values in $[l,u]$ is $[f(l),f(u)]$, which can be computed automatically by simply evaluating $f$.
 
For the regular-expression pruning strategy presented earlier, because concatenation is monotonic (in a sense described later in the paper)
given the partial program $0 \cdot \hole_R$, and intervals $[0, 0]$ and $[\emptyset, (0\mid 1)^*]$, representing sets of possible completions of the (partial) regular expressions $0$ and $\hole_R$, the set of completions of $0 \cdot \hole_R$ can be exactly captured by the interval $[0\cdot \emptyset, 0\cdot(0\mid 1)^*]$ $=[\emptyset, 0\cdot(0\mid 1)^*]$---i.e., what AlphaRegex would compute using a specialized algorithm.

A similar argument explains the pruning technique for \texttt{x\,:=\,0}; y\,:=\,$\hole_E$---i.e., the semantics of assignments is also monotonic with respect to its argument (in sense described later).
% %
% Additionally, if we are given a monotonically decreasing function $f$ and an interval of possible inputs $[l,u]$, the tightest interval that encloses the result of evaluating $f$ on values in $[l,u]$ is $[f(u),f(l)]$.
% %
% \twrchanged{
% For instance, for the operator $\neg r$ for complementing a regular expression, the tightest result on interval $[a,ab^*]$ is $[\neg(ab^*), a]$.
% }

Phrasing the pruning approaches as interval-based abstract interpretation unlocks a new opportunity for pruning that had not been identified in prior work.
Existing approaches assume that holes can produce arbitrary programs that yield arbitrary values in the interval $\top$---i.e., the interval that describes every possible value.
Instead, we show how the order under which the semantics is monotonic can be used to automatically compute tighter bounds on the values of such intervals.
For example, if a hole can only be completed with numbers in the grammar $N\rightarrow 0 \mid 1+N$, our approach can prove that any term derived from nonterminal $N$ must be non-negative.

\subsubsection*{The \semgus Framework}
In practice, it is challenging to build a domain-agnostic synthesizer that can accommodate the diversity of synthesis tasks described above.
To achieve such generality, we target problems in the \semgus format \cite{kim2021semantics}, a unifying domain-agnostic and solver-agnostic framework for specifying arbitrary synthesis problems.
\semgus is expressive enough to specify synthesis problems involving regular expressions, CSV schemas, bitvectors, and imperative programs.
However, to the best of our knowledge, prior \semgus solvers cannot tackle such synthesis problems with reasonable efficiency because of the challenges of generalizing domain-agnostic synthesis techniques beyond na\"{i}ve enumeration.

\subsubsection*{Contributions}
Our work makes the following contributions. 
\begin{itemize}
  \item
    We unify a number of domain-specific pruning approaches into a general \textit{framework for interval-based pruning} (\Cref{sec:overview,sec:synth-procedure}).
  \item
    We define a theory of \textit{semantic monotonicity} that provides sufficient conditions under which it is possible to automate precise interval-based pruning (\Cref{sec:monotonicity}). 
  \item 
    We combine our abstraction framework with a technique called grammar-flow analysis~\cite{gfa} to
    obtain
    \textit{precise hole abstractions} that can be used to further prune the explored search space (\Cref{sec:precise-abstractions}).
  \item
    We provide a technique for automatically \textit{synthesizing orders} under which a semantics is monotonic, thus enabling interval-based pruning (\Cref{sec:order-synth}).      
  \item
    We implemented our framework in the tool \toolname, which supports synthesis problems expressed in \semgus, a domain-agnostic framework for specifying synthesis problems.
   \toolname can automatically discover interval-based pruning approaches that were hard-coded in existing domain-specific solvers, and use them to outperform a vanilla enumeration on several \semgus benchmarks (\Cref{sec:evaluation}).
\end{itemize}

\section{Overview}
\label{sec:overview}

This section illustrates how our framework unifies the pruning techniques used in AlphaRegex~\cite{lee2016alpharegex} (for regular expressions) and SIMPL~\cite{so2017synthesizing} (for imperative programs).
We will further discuss in \Cref{sec:related-work} how our framework also unifies some prior approaches for synthesizing SQL queries~\cite{wang2017scythe}, Datalog programs~\cite{si2018datalog}, and data-processing tasks~\cite{video-trajectories}.
% \subsection{Synthesis from Examples}

% In the rest of this section, we assume that a synthesis problem is given as a regular tree grammar $\grammar$ of possible programs, a semantics $\den\cdot$ that allows us to evaluate the programs in $L(\grammar)$, and a set of input/output examples $\examples=\{(i_1,o_1), \ldots, (i_n,o_n)\}$.
% % 
% A solution is a program $p$ in $\lang(\grammar)$ that is consistent with the examples---i.e., $\forall j. \den{p}(i_j)=o_j$.

The tool AlphaRegex synthesizes regular expressions in the following fixed grammar $G_\alpha$:
\begin{equation*}
    R ::=\ \rc c \mid \epsilon \mid \emptyset \mid (R \mid R) \mid (R\cdot R) \mid R^*
\end{equation*}
Examples are given as pairs of strings and Boolean values, denoting whether a string should be accepted or rejected by the regular expression to be synthesized.
For example, given the examples $\examples^R_1=\{(1,\textit{true}),(10,\textit{true}),(111,\textit{true}),(0,\textit{false}),(00,\textit{false}),(100,\textit{false})\}$, AlphaRegex might synthesize the regular expression $(1\cdot(0\mid 1))^*\cdot 1$, which accepts all non-empty bitstrings where every odd position contains the digit 1.

The tool SIMPL synthesizes imperative programs.
For illustrative purposes, we do not consider programs containing loops, and assume that SIMPL targets programs in the following fixed grammar $\grammar_I$ where the only two variables are $x$ and $y$:
\begin{align*}
    S &::= x := E \mid y := E \mid S ; S \qquad  \qquad 
    E ::= 0 \mid 1 \mid x \mid y \mid E + E \mid E - E
\end{align*}
Examples are given as pairs of input and output states,
where a state is an assignment of values
to $x$ and $y$.
For example, given the examples $\examples^I_1=\{(\langle x=4, y=2\rangle, \langle x=2, y=4\rangle),(\langle x=3, y=3\rangle, \langle x=3, y=3 \rangle)\}$, SIMPL might synthesize the imperative program
\texttt{x\,:=\,x-y; y\,:=\,x+y; x\,:=\,y-x;}, which swaps the values of variables \texttt{x} and \texttt{y}.

\subsection{Top-down Enumeration and Pruning in AlphaRegex and Simpl}
Given a synthesis problem, \textit{top-down enumeration} exhaustively searches over the space of programs in the grammar $\grammar$, expanding \textit{partial programs} according to the productions of the grammar.
A partial program is a tree that can contain \textit{hole} symbols---i.e., unexpanded nonterminals---to be filled by as-yet-undetermined sub-trees.
For example, $0\cdot \hole_R$ and $(\hole_R \mid \hole_R)^*$ are partial programs (regular expressions) that could be further expanded using grammar $\grammar_\alpha$, and \texttt{x\,:=\,x+$\hole_E$} and \texttt{x\,:=\,$\hole_E$; y\,:=\,$\hole_E$} are partial (imperative) programs that could be further expanded using grammar $\grammar_I$.

Blindly enumerating all possible programs is impractical, but through clever pruning strategies, one can mitigate the combinatorial nature of exhaustive enumeration and reach more complex programs deeper in the search space.
Given a partial program $P$, if one can prove that there exists \emph{no} way to turn $P$ into a complete program (i.e., that contains no holes) that is consistent with the given examples $\examples$, the partial program $P$ can be pruned.
When a partial program $P$ is pruned, none of the potentially infinitely many completions of $P$ will ever be considered by the enumeration!
 
\subsubsection*{Pruning in AlphaRegex}
Given the examples $\examples^R_1 = \{ (1,\textit{true})$, $(10,\textit{true})$,  $(111,\textit{true})$, $(0,\textit{false})$, $(00,\textit{false})$,$(100,\textit{false}) \}$, AlphaRegex eventually enumerates the partial program $0 \cdot \hole_R$.
The key observation presented in the AlphaRegex paper is that no matter what program  is used to fill $\hole_R$, the resulting overall program will only accept strings that start with 0, and thus will never accept the string $1$---in particular, the program will be inconsistent with the example $(1,\textit{true})$.
AlphaRegex automates this check by observing that the semantics of regular expressions is such that when we replace $\hole_R$ with the regular expression $(0\mid 1)^*$, we obtain the ``most permissive'' possible regular expression---i.e., if $0 \cdot (0\mid 1)^*$ does not accept all the positive examples, no completion of $0 \cdot \hole_R$ will.
A similar check is performed for negative examples---i.e., if $0 \cdot \emptyset$ does not reject all the negative examples, no completion of $0 \cdot \hole_R$ will.
Therefore, when either of these two checks fails, the corresponding candidate partial program can be pruned away.

\subsubsection*{Pruning in SIMPL}
Given the examples $\examples^I_1=\{(\langle x=4, y=2\rangle, \langle x=2, y=4\rangle),(\langle x=3, y=3\rangle, \langle x=3, y=3 \rangle)\}$, SIMPL eventually enumerates the partial program \texttt{x\,:=\,0; y\,:=\,$\hole_E$}.
The key idea in SIMPL is that no matter what sub-tree is used to replace $\hole_E$, the resulting program is incorrect because it must produce a final state in which $x=0$.
SIMPL automates this check using interval-based abstract interpretation to overapproximate the set of values any program constructed from $\hole_E$ could return.
Intuitively, by applying appropriate interval semantics to the program \texttt{x\,:=\,0; y\,:=\,}$[-\infty, \infty]$, we know that the output state can be overapproximated by the abstract state $(x \in [0,0], y \in [-\infty, \infty]) = [0,0] \times [-\infty, \infty]$.
Because none of the desired output states---e.g., $\langle x=2, y=4\rangle$---fall in this set, this partial program can be pruned.
% 
% The idea of using abstract interpretation to overapproximate the possible values that some completion of a partial program can produce
% is used in other synthesizers~\cite{isil} \rahul{what paper were you thinking here?}, but here we are particularly interested in the fact that SIMPL implements an interval-based static analysis for the language $\grammar_I$---i.e., SIMPL contains a \textit{manually-defined} interval abstract transformer for every construct in the language.
For the interval-based static analysis used in SIMPL, a \textit{manually-defined} interval abstract transformer was created for every construct in $\grammar_I$.

\subsection{A Unifying Framework for Interval-Based Pruning}
The pruning approaches adopted by AlphaRegex and SIMPL are both extremely effective, but require domain-specific insights or manually-designed static analyses to compute precise abstractions.
While the methods for pruning in these two examples appear to be very different, once we describe their semantics in a logical way (in our case, as Constrained Horn Clauses in the  \semgus framework), it becomes possible to handle these pruning approaches in a unified way.
Specifically, both pruning approaches be explained and generalized as instances of \textit{interval-based abstract interpretation}.
Most importantly, we demonstrate that the semantics ascribed to the two grammars $\grammar_\alpha$ and $\grammar_I$ enjoy special properties that allow the appropriate abstract interval-transformers to be created automatically from the user-provided logical semantics.

% The key commonality between the aforementioned pruning approaches is that they both compute an interval-based abstraction of a partial program according to their semantics and, because of properties of the constructs in the grammar, such abstractions can be computed precisely without manually-designed abstract transformers.

\subsubsection*{Interval-Based Pruning in AlphaRegex}
Given the partial program $0 \cdot \hole_R$, our first observation is that one can think of AlphaRegex as computing a precise interval-based abstraction of the partial program $0 \cdot \hole_R$ by interpreting $\hole_R$ as the interval $[\emptyset, (0\mid 1)^*]$---i.e., with the range of all possible strings that a regular expression that fills $\hole_R$ could produce.
% by replacing the $\hole_R$ with the interval $[\emptyset, (0\mid 1)^*]$---i.e., with the range of all possible regular expressions one can produce.
% 
To reformulate this computation in abstract terms:
AlphaRegex is computing an abstract semantics as $\aden{0 \cdot \hole_R}(\hole_R=[\emptyset, (0\mid 1)^*])=[0,0] \cdot^\sharp [\emptyset, (0\mid 1)^*]$, where $\cdot^\sharp$ is the interval abstract transformer for concatenation.
The second key observation is that the abstract value $[0,0] \cdot^\sharp [\emptyset, (0\mid 1)^*]$ can be computed as $[0\cdot \emptyset, 0\cdot (0\mid 1)^*]$---i.e., the abstract transformer $\cdot^\sharp$ over an interval of regular expressions can be computed by applying the concrete operation $\cdot$ to the left and right bounds of its interval arguments.
This last step is actually true for all regular-expression operators!
Our final key observation is that this trivial computation of an interval-abstract transformer can always be performed as long as the semantics of the operator under consideration is \textit{monotonic} with respect to the order over which intervals are defined.
In this case, because intervals are ordered by language inclusion, we have that if $L(r_1)\subseteq L(r_1')$ and $L(r_2)\subseteq L(r_2')$ then $L(r_1\cdot r_2) \subseteq L(r_1'\cdot r_2')$---i.e., the semantics of concatenation is monotonically increasing in its arguments.
% ---i.e., because on a string $s$, we have $\den{r}(s)\in \{0,1\}$, the order is defined as $r_1\sqsubseteq_R r_1'$ if for every string $s$, we have that $\den{r_1}(s)\leq \den{r_1'}(s)$---and the concatenation operator is monotonic with respect to this order in both of its arguments---i.e., if $r_1\sqsubseteq_R r_1'$ and $r_2\sqsubseteq_R r_2'$ then $r_1\cdot r_2 \sqsubseteq r_1'\cdot r_2'$.
Consequently, the abstract transformer for $\cdot^\sharp$ is $[r_l,r_u] \cdot^\sharp [r_l',r_u']= [r_l \cdot r_l',r_u \cdot r_u']$.
The same argument applies to other regular-expression operators.
% ---i.e., $[r_l,r_u]^{*^\sharp}= [r_l^*,r_u^*]$ and $[r_l,r_u]\mid^\sharp [r_l',r_u']= [r_l\mid r_l',r_u\mid r_u']$.
 
The above argument gives us a systematic way to explain how AlphaRegex merely computes an interval abstraction of all possible strings that can be produced by some completion of a partial regular expression.
Thus, if we can determine that the operations in the user-defined semantics in the \semgus format exhibit the monotonicity property above, we automatically get abstract interval transformers for these operations for free.

\subsubsection*{Interval-Based Synthesis for Imperative Program}
It is now easy to see how AlphaRegex and SIMPL are similar.
% 
% Recall the example from earlier of synthesizing the ``swap'' function, and consider the partial program $x\,:=\,0; y\,:=\,\hole_E$.
%
In SIMPL, the abstract semantics of $\hole_E$ for a specific input example (let's say $\langle x=4, y=2\rangle$) is expressed as the interval $[-\infty, \infty]$, which intuitively states that on the given input, if we were to run any of the programs with which one can replace $\hole_E$ on the input example, we would get an output in the interval $[-\infty, \infty]$.
The key point is that although SIMPL provided manually crafted abstract transformers to evaluate the semantics of the partial program for this interval domain, in many cases, such a semantics can be computed automatically,
like we did above!
In particular, all the operators appearing in the partial program $x\,:=\,0; y\,:=\,\hole_E$ are monotonic (in a sense that we will formalize later in the paper), and therefore $\aden{x\,:=\,0; y\,:=\,\hole_E}(x=4, y=2,\hole_E=[-\infty,\infty])$ can be computed as $[ \den{x\,:=\,0; y\,:=\,\hole_E}(x=4, y=2,\hole_E=-\infty),\den{x\,:=\,0; y\,:=\,\hole_E}(x=4, y=2,\hole_E=\infty) ]$.

% This section showed how two existing pruning strategies can be unified as abstraction-based approaches based on the theory of intervals.
% % 
% Most importantly, the unified approach gives us conditions under which the approach can be automated---i.e., when the semantics is monotonic.
% % 
% In Section~\ref{sec:monotonicity}, we formalize the intuition described so far and present an algorithm for automating this form of reasoning.
% 

% % 
% In particular, we show how the recently proposed synthesis framework \semgus exposes the semantics of the programming language for which we are performing synthesis in a way that allows us to reason about its monotonicity.

% Firstly, the semantics for the term $x := E$ takes the input tuple $(x,y)$ and maps it to $(e,y)$, where $e$ corresponds to the value of the sub-AST for the nonterminal $E$.
% %
% Thus, the abstraction representing the program $x := 0$ would be an integer over tuple states $[(0, -\infty), (0, \infty)]$.
% %
% Similarly, the interval representing $y := \hole_E$ would be $[(-\infty, -\infty), (\infty, \infty)]$. 
% %
% Finally, the semantics corresponding to the sequential operator is to first execute the first statement, and then execute the second one.
% %
% This gives the final interval representation $[(0, -\infty), (0, \infty)] = [0,0] \times [-\infty, \infty]$, which is the same result we got earlier.
% %
% Since we expected the output $(x=2, y=4)$ on input $(x=4, y=2)$, and the output is outside of our abstract interval, we can discard this partial program and thus any of the infinite programs that are derived from it.

\subsection{Computing Precise Hole Abstractions via Grammar Flow Analysis}
\label{sec:csv:precise}
% \loris{if space is eventually needed, we can remove this subsection since in practice it doesn't help and we describe it later anyway. Alternatively can become a shorter paragraph with the same example we show in sec 4.5 }
We now illustrate how the interval-based framework unlocks another pruning opportunity.
% 
% The approach we described earlier generates intervals by replacing the holes in the partial programs with the smallest and largest elements of their respective domains---i.e., $[\emptyset, (0\mid 1)^*]$ for regular expressions and $[-\infty, \infty]$ for the possible output values of an expression.\footnote{Note that smallest and largest are defined here with respect to the same order for which the semantics of the underlying language is monotonic.}
% %
% While these ``default'' abstractions are always sound overapproximations of how the set of programs derivable from a nonterminal may behave, sometimes they can be more imprecise than necessary.
% 
Consider the task of synthesizing a regular expression in the following grammar $\grammar_{CSV}$, which defines regular expressions for describing the format of rows in a CSV file---i.e., comma-separated entries that (for the sake of this example) can contain either alphabetic or numerical characters.
\begin{align*}    
    \RowId & ::= \AlphaId \mid \NumId \mid \AlphaId \cdot {,} \cdot \RowId\mid \NumId \cdot {,} \cdot \RowId \\    
    \AlphaId & ::= a \mid \cdots \mid z \mid (\AlphaId \cdot \AlphaId) \mid (\AlphaId \mid \AlphaId) \mid (\AlphaId)^* \\
    \NumId & ::= 0 \mid \cdots \mid 9 \mid (\NumId \cdot \NumId) \mid (\NumId \mid \NumId) \mid (\NumId)^*
\end{align*}

Suppose we are given the set of examples $\examples_{CSV}=\{(\texttt{"303, name"}, \textit{true})\}$ and eventually enumerate the partial program $\hole_{\AlphaId} \cdot {,} \cdot \hole_{\RowId}$.
If one tries to prune this partial program by replacing $\hole_{\AlphaId}$ and $\hole_{\RowId}$ with the ``default'' overapproximation 
$[\emptyset, .^*]$ (where $.$ denotes any character), we would get the following interval to represent the set of possible solutions: $[\emptyset, (.^*\cdot {,} \cdot .^*)]$, and conclude that the partial program $\hole_{\AlphaId}\cdot {,} \cdot \hole_{\RowId}$ cannot be pruned.
A careful analysis of the nonterminal $\AlphaId$ shows that any program derived from $\hole_{\AlphaId}$ can only match alphabetic strings.
Thus, the interval $[\emptyset, (a\mid \cdots \mid z)^*]$ would provide a better abstraction of the semantics of all programs derivable from $\AlphaId$ and allow us to prune the partial program $\hole_{\AlphaId}\cdot {,} \cdot \hole_{\RowId}$ (because no completion can start with ``$\texttt{303}$.'')

Our interval-based framework opens up a simple way to compute these more precise abstractions automatically. 
In particular, once we have a concrete representation of the interval ordering relation, we can use a technique called Grammar Flow Analysis (GFA)~\cite{gfa}.

Specifically, for every nonterminal $N$ we can construct a system of constraints based on the provided grammar and computed abstract semantics, for which the least solution is the tightest interval that overapproximates the semantics of the programs derivable from $N$.
For instance, for the nonterminal $\AlphaId$, the constraints contain as free variables the interval bounds $[l, u]$  we are looking for (in this case, denoted by regular expressions themselves) and take the following form:
\changed{\begin{equation}
\label{eq:gfa-reg-ex-ill}
\forall v. \big(a {\preceq} v {\preceq} a \lor \cdots \lor z {\preceq} v {\preceq} z \lor (l \cdot l) {\preceq} v {\preceq} (u \cdot u) \lor (l \mid l) {\preceq} v {\preceq} (u \mid u) \lor (l)^* {\preceq} v {\preceq} (u)^* \Rightarrow l {\preceq} v {\preceq} u \big)
\end{equation}}
Assuming that $\preceq$ denotes language inclusion, the assignment $l=\emptyset$ and $u=(a\mid \cdots \mid z)^*$ is a valid solution to \Cref{eq:gfa-reg-ex-ill}. We will show in \Cref{sec:precise-abstractions} how we solve such equations.
\section{Top-down Enumeration and Abstraction-based Pruning}
\label{sec:synth-procedure}

In this section, we formulate the synthesis problems that we tackle (\Cref{sec:synthesis-problem}), formalize a top-down enumeration algorithm (\Cref{sec:topdown-enumeration}), and show how interval-based abstractions can be used to prune the set of programs enumerated via top-down enumeration (\Cref{sec:pruning:using:abstractions}). 
We illustrate all our techniques using a simple imperative programming language (\Cref{fig:imp-example}).
\iflabeldefined{sec:app:regex}{\Cref{fig:regex-example}}{Figure 5 in Appendix A} contains another example for a language of regular expressions.

% \loris{you need to say somewhere in this sec that this is a reasonable fragment of semgus and what it misses. Say we don't provide actual semgus formalization as CHCs because it would make theory cumbersome. There is currently a comment before def 3.5, but I would add something here maybe} \rahul{fixed?}

\subsection{Synthesis Problem}
\label{sec:synthesis-problem}
In this section, we introduce the terminology used throughout the rest of the paper, and define the scope of our example-based synthesis problem.
% 
% We start by defining regular tree grammars, which are used to describe sets of programs.
\begin{definition}[Regular tree grammar]
    \label{def:RTG}
    A \emph{regular tree grammar} (RTG) is a tuple {$G = (\mathcal{N}, \Sigma, S, \delta)$},
    where
    $\mathcal{N}$ denotes a finite set of non-terminal symbols;
    $\Sigma$ is a ranked alphabet;
    $S \in \mathcal{N}$ is the start nonterminal;
    and $\delta$ is a finite set of productions $N_0 \to \rho(N_1, \ldots, N_n)$.
    % , where
    % if $a(\rho) = (\tau_1, \ldots, \tau_n)$ for $\rho \in \Sigma$, then for all $i \in [1..n]$, $a(N_i) = \tau_i$.
\end{definition}

% \lorischanged{
For each nonterminal $N$, we use $\hole_N$ to denote a node variable---i.e., a hole---associated with nonterminal $N$.
We define a partial program $P$ to be a tree that may contain holes.
Given a partial program $P$ with a leftmost hole occurrence $\hole_N$, we say that a program $P'$ can \textit{be derived in one step} from $P$ if there exists a production $N \to \rho(N_1, \ldots, N_n)$ such that replacing the leftmost hole occurrence $\hole_N$ with  $\rho(\hole_{N_1}, \ldots, \hole_{N_n})$ results in $P'$---denoted by $P \mapsto P'$.
We say that a partial program $P'$ is \textit{derived} from another partial program $P$ in zero or more steps, denoted by $P \mapsto^* P'$, if it is in the reflexive and transitive closure of the one-step derivation relation.

Given an RTG $G$, we say that a partial program $P$ is \emph{generated by a nonterminal $N$} if it can be generated from $\hole_n$---i.e., $\hole_n\mapsto^* P'$.
Finally, program $P$ is \emph{complete} if it does not contain holes.
We use $\LL(N)$ to denote the set of complete programs that can be generated by a nonterminal $N$.
We use $\LL(G)=\LL(S)$ to denote the set of complete programs accepted by the grammar $G$.
% }

\begin{example}[Imperative Grammar]
Consider the imperative language defined by grammar $G_I$ in \Cref{fig:imp-grammar}.
From the partial program $\texttt{x\,:=\,0; y\,:=\,} \hole_E$ we can derive in one step the complete program $\texttt{x\,:=\,0; y\,:=\,x}$ and the partial program $\texttt{x\,:=\,0; y\,:=\,} (\hole_E + \hole_E)$.
\end{example}

The following definition shows how to associate a semantics to programs defined by a grammar.
\begin{definition}[Semantics] 
\label{def:semantics}
    Let {$G = (\mathcal{N}, \Sigma, S, \delta)$} be a regular-tree grammar with set of non-terminals $\mathcal{N}=\{N_1, \ldots N_k\}$.
    A \emph{semantics} for $G$ is a pair $(\{\dennon{N_1}{\cdot}, \ldots, \dennon{N_k}{\cdot}\},\sigma)$, such that each $\dennon{N_i}{t}:t \in \LL(N_i) \rightarrow I_{N_i}\rightarrow O_{N_i}$ is a function symbol for which the interpretation is given 
    by $\sigma$, a function that maps each production $M_0 \to \rho(M_1, \ldots, M_n)$ to a set of well-typed first-order implication formulas---i.e., the rules that define the function $\dennon{M_0}{t}$---of the following form:
    \begin{equation}
    \label{Eq:CHCSchema}
    \mbox{
    \begin{minipage}{\textwidth}
    \begin{prooftree}
        \AxiomC{$\dennon{M_1}{t_1}({x_1}) = {y_1} \quad \cdots \quad \dennon{M_n}{t_n}({x_n}) = {y_n}$}
        \AxiomC{$\varphi_=(x, x_1, \ldots, x_n, y_1, \ldots, y_n) \land {y} = f({x}, {y_1}, \ldots, {y_n})$}
        \RightLabel{}
        \BinaryInfC{$\den{\rho(t_1, \ldots, t_n)}({x}) = {y}$}
    \end{prooftree}
    \end{minipage}
    }
    \end{equation}
    In the formulas, all variables $t_i, x_i, y_i$ are universally quantified.
    $\varphi_=$ is a conjunction of equalities over the variables $x$, $x_i$, and $y_i$; and $f$ is a function.
    %
    % \rahulchanged{The types of $\den{\cdot}$ have to match compositionally: for any nonterminal $M \to \rho(M_1, \ldots, M_N)$, if $\den{\cdot}_{M}$ has type $(\tau_1, \ldots, \tau_n)$, then $\den{\cdot}_{M_i}$ must have type $\tau_i$.}

    The interpretation of the semantic function symbols $\{\dennon{N_1}{\cdot}, \ldots, \dennon{N_k}{\cdot}\}$ is then the least fixed-point solution of the set of first-order formulas $\bigcup_{p\in \delta} \sigma(p)$.
\end{definition}
These first-order formulas used to define the semantics are a restricted form of \textit{Constrained Horn Clauses} (CHCs)~\cite{spacer} and we will refer to formulas of the above form as CHCs.
The above definition of semantics is a restricted fragment of the \semgus framework, which instead supports CHCs with arbitrary predicates.
Our restricted format cannot capture nondeterminism, but allows us to model executable semantics for deterministic programs, which are the focus of our work.
In the rest of the paper, to avoid clutter we drop the nonterminal subscripts (unless noted otherwise) and simply write $\den{\cdot}$ to denote all semantic relations. 
When writing $\den{\cdot}$, we also implicitly assume the corresponding mapping $\sigma$ is given to us, so we drop it from most definitions.

\begin{figure}
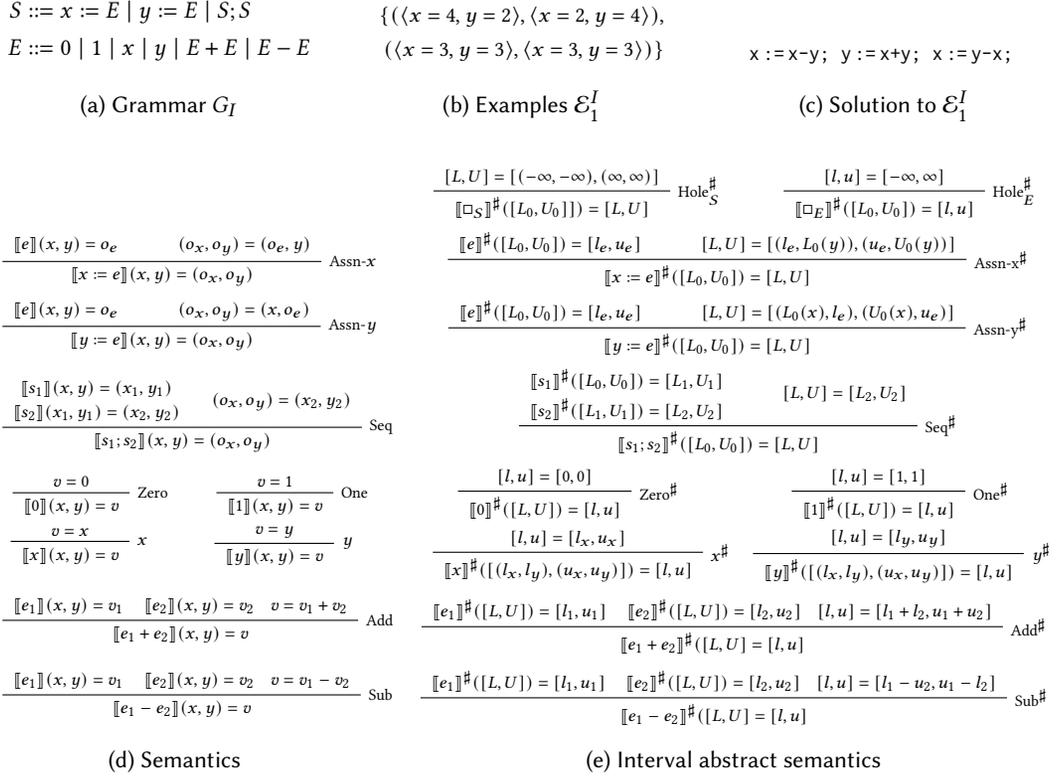

    % \label{fig:regex}
    \begin{subfigure}{0.3\textwidth}
        \small{
            \begin{align*}
                S &::= x := E \mid y := E \mid S ; S \\
                E &::= 0 \mid 1 \mid x \mid y \mid E + E \mid E - E
            \end{align*}
        \caption{Grammar $G_I$}
        \label{fig:imp-grammar}
        }
    \end{subfigure}
    \hfill
    \begin{subfigure}{0.3\textwidth}
        \footnotesize{
            \begin{align*}
                \{(\langle x=4, y=2\rangle, \langle x=2, y=4\rangle),\\
                (\langle x=3, y=3\rangle, \langle x=3, y=3 \rangle)\}
            \end{align*}
        \caption{Examples $\examples^I_1$}
        \label{fig:imp-constraint}
        }
    \end{subfigure}
    \hfill
    \begin{subfigure}{0.3\textwidth}
        \footnotesize{
            \centering
            \texttt{x\,:=\,x-y; y\,:=\,x+y; x\,:=\,y-x;}
            \vspace{4pt}
        \caption{Solution to $\examples^I_1$}
        \label{fig:imp-solution}
        }
    \end{subfigure}
    \vspace{14pt}

    \begin{subfigure}{0.33\textwidth}
        \tiny{
        % \begin{tikzpicture}
        %     \node[text width=7cm,draw,inner sep=0.6em](mybox){
                \centering
                \begin{prooftree}
                    \AxiomC{$\den{e}(x,y) = o_e$}
                    \AxiomC{$(o_x, o_y) = (o_e, y)$}
                    \RightLabel{Assn-$x$}
                    \BinaryInfC{$\den{x := e}(x,y) = (o_x, o_y)$}
                \end{prooftree}
                % \vspace{\proofspacing}
                \begin{prooftree}
                    \AxiomC{$\den{e}(x,y) = o_e$}
                    \AxiomC{$(o_x, o_y) = (x, o_e)$}
                    \RightLabel{Assn-$y$}
                    \BinaryInfC{$\den{y := e}(x,y) = (o_x, o_y)$}
                \end{prooftree}
                \vspace{\proofspacing}
                \begin{prooftree}
                    \AxiomC{\stackanchor{$\den{s_1}(x,y) = (x_1, y_1)$}{$\den{s_2}(x_1, y_1) = (x_2, y_2)$}\qquad$(o_x, o_y) = (x_2, y_2)$}
                    \RightLabel{Seq}
                    \UnaryInfC{$\den{s_1 ; s_2}(x,y) = (o_x, o_y)$}
                \end{prooftree}
                \vspace{\proofspacing}
                \centering
                \begin{minipage}{0.35\textwidth}
                \begin{prooftree}
                    \AxiomC{$v = 0$}
                    \RightLabel{Zero}
                    \UnaryInfC{$\den{0}(x,y) = v$}
                \end{prooftree}
                \end{minipage}
                \hspace{1cm}
                \centering
                \begin{minipage}{0.35\textwidth}
                \begin{prooftree}
                    \AxiomC{$v = 1$}
                    \RightLabel{One}
                    \UnaryInfC{$\den{1}(x,y) = v$}
                \end{prooftree}
                \end{minipage}
                \vspace{\proofspacing}
                % \vspace{\proofspacing}
                %
                \centering
                \begin{minipage}{0.35\textwidth}
                \begin{prooftree}
                    \AxiomC{$v = x$}
                    \RightLabel{$x$}
                    \UnaryInfC{$\den{x}(x,y) = v$}
                \end{prooftree}
                \end{minipage}
                \hspace{1cm}
                \centering
                \begin{minipage}{0.35\textwidth}
                \begin{prooftree}
                    \AxiomC{$v = y$}
                    \RightLabel{$y$}
                    \UnaryInfC{$\den{y}(x,y) = v$}
                \end{prooftree}
                \end{minipage}
                \vspace{\proofspacing}
                \begin{prooftree}
                    \AxiomC{$\den{e_1}(x,y) = v_1$ \quad $\den{e_2}(x,y) = v_2 \quad v = v_1 + v_2$}
                    \RightLabel{Add}
                    \UnaryInfC{$\den{e_1 + e_2}(x,y) = v$}
                \end{prooftree}
                \vspace{\proofspacing}
                \begin{prooftree}
                    \AxiomC{$\den{e_1}(x,y) = v_1$ \quad $\den{e_2}(x,y) = v_2 \quad v = v_1 - v_2$}
                    \RightLabel{Sub}
                    \UnaryInfC{$\den{e_1 - e_2}(x,y) = v$}
                \end{prooftree}
                \vspace{\proofspacing}
    %         };
    %         \node[text=gray,anchor=west,fill=white,xshift=0.5em] at (mybox.north west)
    %          {\textsc{Imperative Program Semantics}};
    % \end{tikzpicture}
        \caption{Semantics}
        \label{fig:imp-sem}
    }
    \end{subfigure}
    \hfill
    \begin{subfigure}{0.60\textwidth}
        \tiny{
            \centering            
            \begin{minipage}{0.45\textwidth}
                \begin{prooftree}
                    \AxiomC{$[L,U] = [(-\infty, -\infty), (\infty, \infty)]$}
                    \RightLabel{Hole$^\sharp_S$}
                    \UnaryInfC{$\dena{\hole_S}([L_0,U_0]]) = [L,U]$}
                \end{prooftree}
            \end{minipage}
            \hspace{0.8cm}
            \begin{minipage}{0.4\textwidth}
                \begin{prooftree}
                    \AxiomC{$[l,u] = [-\infty, \infty]$}
                    \RightLabel{Hole$^\sharp_E$}
                    \UnaryInfC{$\dena{\hole_E}([L_0,U_0]) = [l,u]$}
                \end{prooftree}
            \end{minipage}
                \begin{prooftree}
                    \AxiomC{$\dena{e}([L_0,U_0]) = [l_e, u_e]$}
                    \AxiomC{$[L,U] = [(l_e, L_0(y)), (u_e,U_0(y))]$}
                    \RightLabel{Assn-x$^\sharp$}
                    \BinaryInfC{$\dena{x := e}([L_0,U_0]) = [L,U]$}
                \end{prooftree}
                \begin{prooftree}
                    \AxiomC{$\dena{e}([L_0,U_0]) = [l_e, u_e]$}
                    \AxiomC{$[L,U] = [(L_0(x), l_e), (U_0(x), u_e)]$}
                    \RightLabel{Assn-y$^\sharp$}
                    \BinaryInfC{$\dena{y := e}([L_0,U_0]) =[L,U]$}
                \end{prooftree}
                \begin{prooftree}
                    \AxiomC{\stackanchor{$\dena{s_1}([L_0,U_0]) = [L_1,U_1]$}{$\dena{s_2}([L_1,U_1]) = [L_2,U_2]$}}
                    \AxiomC{$[L,U] = [L_2,U_2]$}
                    \RightLabel{Seq$^\sharp$}
                    \BinaryInfC{$\dena{s_1 ; s_2}([L_0,U_0]) = [L,U] $}
                \end{prooftree}
                \begin{minipage}{0.35\textwidth}
                \begin{prooftree}
                    \AxiomC{$[l,u] = [0,0]$}
                    \RightLabel{Zero$^\sharp$}
                    \UnaryInfC{$\dena{0}([L,U]) = [l,u]$}
                \end{prooftree}
                \end{minipage}
                \hspace{1.4cm}
                \centering
                \begin{minipage}{0.35\textwidth}
                \begin{prooftree}
                    \AxiomC{$[l,u] = [1,1]$}
                    \RightLabel{One$^\sharp$}
                    \UnaryInfC{$\dena{1}([L,U]) = [l,u]$}
                \end{prooftree}
                \end{minipage}
                
                \begin{minipage}{0.45\textwidth}
                \begin{prooftree}
                    \AxiomC{$[l,u] = [l_x,u_x]$}
                    \RightLabel{$x^\sharp$}
                    \UnaryInfC{$\dena{x}([(l_x,l_y), (u_x,u_y)]) = [l,u]$}
                \end{prooftree}
                \end{minipage}
                \hspace{0.4cm}
                \centering
                \begin{minipage}{0.45\textwidth}
                \begin{prooftree}
                    \AxiomC{$[l,u] = [l_y,u_y]$}
                    \RightLabel{$y^\sharp$}
                    \UnaryInfC{$\dena{y}([(l_x,l_y), (u_x,u_y)]) = [l,u]$}
                \end{prooftree}
                \end{minipage}
                \begin{prooftree}
                    \AxiomC{$\dena{e_1}([L,U] ) = [l_1, u_1]$ \quad $\dena{e_2}([L,U]) = [l_2, u_2] \quad [l,u] = [l_1 + l_2, u_1 + u_2]$}
                    \RightLabel{Add$^\sharp$}
                    \UnaryInfC{$\dena{e_1 + e_2}([L,U] = [l,u]$}
                \end{prooftree}
                \begin{prooftree}
                    \AxiomC{$\dena{e_1}([L,U] ) = [l_1, u_1]$ \quad $\dena{e_2}([L,U]) = [l_2, u_2] \quad [l,u] = [l_1 - u_2, u_1 - l_2]$}
                    \RightLabel{Sub$^\sharp$}
                    \UnaryInfC{$\dena{e_1 - e_2}([L,U] = [l,u]$}
                \end{prooftree}                
        \caption{Interval abstract semantics}
        \label{fig:imp-sem-abs}
        }
    \end{subfigure}
    
    \caption{An example-based \semgus problem for imperative programs (\Cref{fig:imp-grammar,fig:imp-sem,fig:imp-constraint}), and a sound abstract semantics for the grammar $G_I$ (\Cref{fig:imp-sem-abs}). We use lowercase $l,u$ variables to denote integers and uppercase variables $L,U$ to denote pairs of integers, where $L(x)$ and $U(x)$ correspond to the first element of the pair, and $L(y), U(y)$ the second element.}
    \label{fig:imp-example}
\end{figure}

\begin{example}[Imperative Semantics]
    \label{ex:imperative-sem}    
    Figure~\ref{fig:imp-sem} presents an (operational) semantics for programs in the grammar $G_I$ (Figure~\ref{fig:imp-grammar}).
    The semantics associated with nonterminal $E$, i.e., $\den{\cdot}_E$, has type $\IntId\times \IntId \rightarrow \IntId$, i.e., the programs derived from $E$ map pairs of integer variable values to integers.
    The semantics $\den{\cdot}_S$ has type $\IntId\times \IntId \rightarrow \IntId \times \IntId$, i.e., a program derived from $S$ maps a pair of values---one for each variable---to a pair of values.
    In each rule, instead of explicitly presenting the formula $\varphi_=$ (as in Eqn.~(\ref{Eq:CHCSchema})),
    we introduce unique names to denote all variables that are equal according to $\varphi_=$.
    In the rule Sub, the function appearing in the premise of the CHC would be $f_{\textit{Minus}}((x,y),v_1,v_2)=v_1-v_2$.
\end{example}

%In the rest of the paper, unless otherwise noted, we consider synthesis problems in which the specification is provided using a set of input-output examples (given a logical specification, one can always generate a set of examples as needed by applying counterexample-guided inductive synthesis (CEGIS)~\cite{lorisbook}).
In the remainder of the paper, we consider synthesis problems in which the program is only required to be correct on a finite set of examples.
% \twr{``... to meet a logical specification for a finite set of input examples'' is an awkward phrase.}
% 
We are now ready to define the synthesis problems we consider in this paper.
Our definition is an instance of the Semantics-Guided Synthesis (\semgus) framework for describing synthesis problem, where semantics are provided using CHCs; thus, we use the same name.
\begin{definition}[Example-based \semgus Problem]
    An \textit{example-based \semgus problem} is a triple $(G,\den{\cdot},\examples)$ where $G$ is a regular-tree grammar specifying the search space, $\den{\cdot}$ is a semantics that uses CHCs to ascribe meaning to the trees/programs generated by the grammar, and $\examples=\{(i_1, o_1), \ldots, (i_n, o_n)\}$ is a set of input-output examples.
 
    A program $P$ in the language $\LL(G)$ described by the grammar $G$ is a \textit{solution} to the synthesis problem if it is consistent with all the examples, i.e., for every $1\leq j\leq n$, we have $\den{P}(i_j)=o_j$, which we denote by $P \vdash \examples$.
\end{definition}

The grammar $G_I$, its semantics, and the set of examples $\examples^I$ shown in \Cref{fig:imp-example} form an example-based \semgus problem.

\subsection{Top-Down Enumeration}
\label{sec:topdown-enumeration}

In this section, we describe the standard top-down enumeration approach to synthesis that we will apply to example-based \semgus problems.
To find a solution to an example-based synthesis problem $(G, \den{\cdot}, \examples)$, top-down enumeration enumerates trees $t \in \LL(G)$ to find one that is consistent with all the examples, i.e., $t \vdash \examples$.

% \begin{wrapfigure}{l}{0.5\textwidth}
    % \centering
\begin{wrapfigure}{R}{0.5\textwidth}
\vspace{-7mm}
{\footnotesize
\begin{minipage}{0.5\textwidth}
\setlength{\intextsep}{1\baselineskip}
    \begin{algorithm}[H]
        \caption{Top-down enumeration w. pruning}
        \label{alg:top-down-synth}
        \begin{algorithmic}[1]
            \Function{Search}{\grammar, $\den{\cdot}$, $\examples$}
                \State $Q \gets \{\hole_{S}\}$ \label{alg:queueinit}
                \While{$Q \neq \emptyset$} \label{alg:loopiter}
                    \State $P \gets \code{Dequeue}(Q)$
                    \If{$\code{IsComplete}(P)\ \land\ P \vdash \examples$} \label{alg:progsuccess}
                        \Return{P}
                    \Else
                        \ForAll{$P'$ s.t. $P \mapsto P'$} \label{alg:allderiv}
                            \If{$\lnot$\Call{Prune}{$P',\examples$}}
                                $\code{Enqueue}(Q,P')$
                            \EndIf
                        \EndFor
                    \EndIf
                \EndWhile
            \EndFunction
            \vspace{4pt}
            \Function{Prune}{$P', \examples$} \label{alg:prune-def}
                \ForAll{$(i_k, o_k) \in \examples$} 
                    \If {$o_k \not \in \dena{P'}(i_k)$} \label{alg:validprune}
                        \Return \code{True}
                    \EndIf
                \EndFor
                \State \Return \code{False}
            \EndFunction \label{alg:prune-end}
        \end{algorithmic}
    \end{algorithm}
    \label{fig:enter-label}
\end{minipage}
}
\vspace{-7mm}
\end{wrapfigure}
% \end{wrapfigure}

% \paragraph{Enumerating by Expanding Partial Programs}
Given a synthesis problem, Algorithm~\ref{alg:top-down-synth} enumerates partial programs by exploring their one-step derivations until a complete program is found that is a solution.
The algorithm starts with the smallest partial program, a single hole $\hole_{\textit{S}}$, corresponding to start nonterminal $S$ (Line~\ref{alg:queueinit}), and iterates over all partial and total programs in the grammar using a priority queue $Q$ (Line~\ref{alg:loopiter}).
The priority for the queue controls the order in which programs are enumerated (e.g., by size, by depth, etc.).

If the program $P$ considered at a certain iteration is complete---i.e., it has no holes---the program is evaluated against the synthesis constraint, and returned as the answer if it satisfies all input-output examples in $\examples$ (Line~\ref{alg:progsuccess}).
If the considered program $P$ is partial---i.e., contains a sequence of holes---the leftmost hole is successively replaced with each possible production to obtain all the programs $P'$ derivable in one step from $P$ (Line~\ref{alg:allderiv}).

The algorithm then queries a \textsc{Prune} function (whose implementation will be explained in Section~\ref{sec:pruning:using:abstractions}) that determines if we can soundly discard partial program $P'$, which would prune away the potentially-infinite set of concrete programs derivable from $P'$.
The following definition states that a function \textsc{Prune} is sound if it only prunes partial programs that cannot derive valid complete solutions to the synthesis problem.
\begin{definition}
\label{def:prune:sound}
We say that the function $\textsc{Prune}$ is \textit{sound} if for every partial program $P'$, if $\textsc{Prune}(P', \examples)=\textsc{True}$ then there exists no complete program $P''$ such that \rone $P''$ is derivable from $P'$, and \rtwo $P''$ is consistent with the examples---i.e., $\neg \exists P''. P'\mapsto^* P'' \wedge P'' \vdash \examples$.
\end{definition}

If the pruning function is sound and the priority queue returns the smallest program with respect to a well-founded order over programs, the algorithm (under reasonable conditions) finds a solution if one exists~\cite{sketch}.
% \begin{theorem}
% \label{thm:completeness}
% \rahulchanged{
% Suppose a total order $\leq$ is given to the set of programs in $\LL(G)$.
% %
% If the priority queue is prioritized \textit{fairly}, i.e., it returns the smallest program with respect to $\leq$ with no infinite descending chains, and the function \textsc{Prune} is sound,} \Cref{alg:top-down-synth} always finds a solution to a given synthesis problem if a solution exists.
% \end{theorem}
% \loris{for proof of this you can point to my book I think}

\subsection{Pruning using Abstractions}
\label{sec:pruning:using:abstractions}

While there are many ways to design pruning functions, in this paper we are only interested in pruning by considering \textit{abstract representations} of certain sets of values.
For a candidate partial program $P'$, we wish to capture the sets of possible outputs of the set of possible completions of $P'$ when these complete programs are run on input $i_k$ from an input-output example $ (i_k, o_k)$ in $\examples$.

Later, we denote such a value-set by $\dena{P'}(i_k)$, which represents an overapproximation of the set of possible states that any program $P$ derived from candidate $P'$ can produce, given input $i_k$.
To define $\dena{P'}$ formally, we need to introduce some terminology from abstract interpretation.
In this paper, we focus on interval-based abstractions---i.e., a set of elements from a set $Y$ is abstracted using a pair/interval $[y_1,y_2]$ of elements to denote the boundaries of the set.

\begin{definition}[Interval abstract domain]
    An \textit{interval abstract domain} for a set $Y$ is a tuple $(Y\times Y,\preceq,\top,\bot)$, where  $\preceq$ is a partial order over the elements of $Y$.
    A concrete element $y \in Y$ is in abstract interval $[y_l, y_u]$ if $y_l \preceq y \preceq y_u$.
    We also define the order $\sqsubseteq$ on $Y\times Y$ as $[y_1,y_2] \sqsubseteq [y_1',y_2']$ iff $y_1'\preceq y_1$ and $y_2\preceq y_2'$. 
    By definition, $\bot \sqsubseteq [y_1,y_2] \sqsubseteq \top$ for any $[y_1,y_2]$.
\end{definition}
%
% We sometimes denote $\bot$ as the empty set $\emptyset$.

To define the value computed by $\dena{P'}$, we introduce a notion of abstract semantics that assigns meanings to both total and partial programs, and lifts the semantics $\den{\cdot}$ to operate over intervals.

We use {$G_{int} = (\mathcal{N}, \Sigma \cup\{\hole_{N_1}, \ldots, \hole_{N_k}\}, S, \delta\cup \{N_i\leftarrow \hole_{N_i} \mid 1\leq i\leq k \})$} to denote the same grammar as $G$ but where each nonterminal can derive a hole, with the intention of defining the semantics over an interval domain.
Note that $\LL(G_{int})$ contains all the partial programs derivable from $G$.
(We will show in Section~\ref{sec:auto-gen-sem} how such domains can be provided.)
\begin{definition}[Interval Abstract Semantics]
    \label{def:abs-sem}
    Let {$G = (\mathcal{N}, \Sigma, S, a, \delta)$} be a regular-tree grammar with a set of non-terminals $\mathcal{N}=\{N_1, \ldots N_k\}$, and let $(\{\dennon{N_1}{\cdot},\ldots,\dennon{N_k}{\cdot}\},\sigma)$ be a semantics as defined in Definition~\ref{def:semantics}.

    An {\textit{interval semantics}} $(\{\denanon{N_1}{\cdot},\ldots,\denanon{N_k}{\cdot}\},\sigma^\sharp)$
    for $G_{int}$---i.e., the grammar that augments $G$ with holes, representing the language of partial programs---is a semantics defined over interval types.
    Each type of $\dena{\cdot}_{N_i}$ is lifted to the corresponding interval type---i.e., if 
    $\den{\cdot}_{N_i} : I_{N_i} \to O_{N_i}$, then $\dena{\cdot}_{N_i} : I_{N_i}\times I_{N_i} \to O_{N_i}\times O_{N_i}$. We assume that each type $T$ has an associated order $\preceq_T$.
    
    The interval semantics for $G_{int}$ is a \emph{sound interval abstract semantics} for $G$ if every (potentially partial) program $P\in\LL(G_{int})$ maps an input interval to an output interval that overapproximates the set of possible outputs that any program $P'$ derivable from $P$ would produce on the same input:
    $$\forall P \in \LL(G_{int}), \forall [l,u], \forall l\preceq x\preceq u, \forall P' \in \LL(G), P \mapsto^* P' \Rightarrow \dennon{S}{P'}(x) \in \denanon{S}{P}([l,u])$$    

    % \loris{Ideally this should defined for every $N$, but we don't have a def of partial programs. }
\end{definition}

We will discuss in Section~\ref{sec:monotonicity} how an abstract semantics can be defined inductively in practice.
Again, to avoid notational clutter we typically write $\dena{\cdot}$ to denote an abstract semantics.
\begin{example}
The semantics defined in \Cref{fig:imp-sem-abs} is a sound interval abstract semantics for $G_I$.
The abstract semantics is defined over intervals of $(x,y)$-pairs: $[(x_l,y_l), (x_u, y_u)]$.
The abstract semantics for nonterminal $E$ has type $\dena{\cdot}_E : (\IntId \times \IntId) \times (\IntId \times \IntId) \to \IntId \times \IntId$.
The abstract semantics for $\dena{\cdot}_S$ has type $(\IntId \times \IntId) \times (\IntId \times \IntId) \to (\IntId \times \IntId) \times (\IntId \times \IntId)$.
The partial order on $(x,y)$-pairs is the pairwise order $(x_1, y_1) \preceq (x_2, y_2)$ iff $x_1 \leq x_2$ and $y_1 \leq y_2$.
An interval's type is a pair of the types of the interval's components.
The abstract semantics of a hole is the widest possible: $\dena{\hole_S}([L,U]) = [(-\infty, -\infty), (\infty, \infty)]$, and $\dena{\hole_E}([L,U]) = [-\infty, \infty]$.
We assume this instantiation of the abstract semantics of holes
for now, but we revisit it in Section~\ref{sec:precise-abstractions}.
\end{example}

To summarize, for a partial program $P$, if for any of the examples $(i_k,o_k)$ in $\examples$, the output $o_k$ is not in $\dena{P}(i_k)=[l_o,u_o]$---i.e., it is not the case that $l_o\preceq_O o_k \preceq_O u_o$---the definition of a sound abstract semantics tells us that $P$ can be \textit{pruned} from further consideration:
$\dena{P}(i_k)$ overapproximates the outputs of all possible completions of $P$, and there only needs to be one failure among all the input-output examples to determine that no completion of $P$ could be a correct answer.

\begin{restatable}[Sound Abstract Semantics and Pruning]{theorem}{pruning}
    \label{thm:approx}
    If $\dena{\cdot}$ is a \emph{sound interval abstract semantics} for $G$, the function \textsc{Prune} described in Algorithm~\ref{alg:top-down-synth} (lines~\ref{alg:prune-def}-\ref{alg:prune-end}) is sound.
\end{restatable}

The following example illustrates how the algorithm can prune using interval abstract semantics.
\begin{example}[Interval-Based Pruning]
    We now show how to compute the interval abstraction for $\dena{\texttt{x\,:=\,0; y\,:=\,} \hole_E}([(4,2),(4,2)])$ for the input-output example $(\langle x=4, y=2 \rangle, \langle x=2, y=4\rangle)$ from our abstract semantics in Figure~\ref{fig:imp-sem-abs}.
    We have $\dena{0}([(4,2),(4,2)]) = [0,0]$ and $\dena{\texttt{x\,:=\,0}}([(4,2)]) = [(0,2), (0,2)]$, 
    as defined in the $\textrm{Zero}^\sharp$ and $\textrm{Ass-x}^\sharp$ rules, respectively.
    The hole $\hole_E$ has abstract semantics $\dena{\hole_E}([(0,2)]) = [-\infty, \infty]$
    (rule $\textrm{Hole}_{E}^\sharp$), which allows us to define $\dena{\texttt{y\,:=\,}\hole_E}([(0,2)]) = [(0, -\infty), (0, \infty)]$ (rule $\textrm{Ass-y}^\sharp$).
    If we combine these results using the
    $\textrm{Seq}^\sharp$ rule,
    we get $\dena{\texttt{x\,:=\,0; y\, :=\,}\hole_E}([(4,2)]) = [(0, \infty), (0, -\infty)]$.
    Because the expected output $(2,4)$ is not in this interval, we can prune this candidate partial program.
\end{example}
\section{Automated Construction of Precise Interval Abstractions}
\label{sec:monotonicity}

To perform pruning in practice, we need to have an abstract semantics in hand to run the abstraction-based synthesis procedure.
Typically, designing precise abstract semantics (even for interval abstract domains) is the topic of complex research papers, and bugs are often found in tools that use interval abstract interpretation~\cite{amurth}.
In this section, we show our main result: if the concrete semantics is monotonic (in a sense that we define formally), a corresponding sound (and precise) interval abstract semantics can be generated automatically from the concrete semantics.
For simplicity, in the rest of this section we assume that every production has only one associated CHC, and describe the points in which modifications need to be made to accommodate multiple CHCs.

\subsection{Abstract Semantics of Monotonic Functions}

\iflabeldefined{sec:app:regex}{In our running examples in \Cref{fig:imp-example} and \Cref{fig:regex-example}}{In our running example in \Cref{fig:imp-example}}, the concrete and abstract semantics looked very similar to each other.
Consider, for example, the concrete and abstract semantics for Union in \iflabeldefined{sec:app:regex}{\Cref{fig:regex-example}}{Figure 5} in \iflabeldefined{sec:app:regex}{\Cref{sec:app:regex}}{Appendix A}.
On the input $s$, the concrete semantics computes the two matrices $A_1 = \den{r_1}(s)$ and $A_2 = \den{r_2}(s)$ for the two subexpressions $r_1$ and $r_2$, and outputs the matrix $A = A_1 + A_2$.
Similarly, the abstract semantics computes two intervals of matrices $[L_1,U_1]$ and $[L_2,U_2]$ for the two subexpressions $r_1$ and $r_2$, and outputs the interval of matrices $[L,U]$ computed as $[L_1+L_2,U_1+U_2]$.

In this example, the two elements of the interval of matrices can be computed by taking the pairwise operation that was applied to the concrete semantics---i.e., the sum of the semantics of the subexpressions.
In other words, the left and right bounds for the interval representing the abstract semantics for this particular production are precisely the concrete semantics applied to the left and right bounds of the input intervals, respectively.
This pattern applies to all the rules except Sub$^\sharp$ in \Cref{fig:imp-example}, where one can still do something analogous.
For the subtraction operator, given an interval of matrices $[l_i, u_i]$ for the abstract semantics of subexpression $e_i$, the abstract semantics of $e_1 - e_2$ can be defined as $[l_1 - u_2, u_1 - l_2]$

However, this recipe may not always result in a sound interval abstract semantics.
Consider a language involving a function $\den{f}(x) = x^2$ that outputs the square of its input: if we wrote $\dena{f}([l,u]) = [l^2, u^2]$, then we would have $\dena{f}([-2,2]) = [4,4]$, which is not a sound abstract semantics---i.e.,  $0$ is in the interval $[-2,2]$, but the corresponding output $f(0)=0$ is not in the interval $[4,4]$. 

Why does the endpoint recipe work for some functions, but not for others? 
The relevant property is the \emph{monotonicity} of the function $f$ appearing in the CHC that defines the concrete semantics (with respect to some order).
In this section, we define what it means for the semantics to exhibit monotonicity, and demonstrate how to automatically check monotonicity of the semantics.

\begin{definition}[Monotonicity]
\label{def:monotonicity}
    Consider a function {$f: Y_0 \times Y_1 \times \ldots \times  Y_n \to Y$, where $Y_0, Y_1, \cdots, Y_n$, and $Y$ are sets that are partially ordered under $\preceq_0, \preceq_1, \cdots, \preceq_n$, and $\preceq$, respectively.}
    For each argument $i$, we define the \textit{freezing function} {$f_i[\vec c]: Y_i \to Y$} produced by fixing a vector of constants $\vec c$ for the arguments aside from {$y_i$---i.e., $f_i[\vec c](y_i) = f(c_0, c_1,...c_{i-1},y_i,c_{i+1},...c_n)$}.
    
    We say that $f$ is \textit{monotonically increasing} ($\monup$) in its $i^{\textit{th}}$ argument if for every $\vec c$, the output of the freezing function $f_i$ \textit{increases} when the $i^{\textit{th}}$ input \textit{increases}:
    $\forall y_1,y_2 \in Y_i, y_1 \preceq_i y_2 \implies f_i[\vec c](y_1) \preceq f_i[\vec c](y_2)$.
    
    Likewise, $f$ is \textit{monotonically decreasing} ($\mondown$) in its $i^{\textit{th}}$ argument if for every $\vec c$, the output of the freezing function $f_i$ \textit{decreases} when the $i^{\textit{th}}$ input \textit{increases}:
        $\forall y_1, y_2 \in Y_i, y_1 \preceq_i y_2 \implies f_i[\vec c](y_2) \preceq f_i[\vec c](y_1)$.

    If for every argument $i$, $f$ is monotonically increasing or decreasing in the $i^{\textit{th}}$ argument, we say that $f$ is \textit{monotone}.
\end{definition}

Before continuing, we observe that if all the order relations $\preceq_i$ and functions in the semantics are expressed in a decidable theory, the problem of checking whether a function is monotone can be solved using a constraint solver---i.e., by checking if the constraints in \Cref{def:monotonicity} hold.

\begin{example}[Monotonic Functions]
\label{ex:monotonicity}
    Recall the semantics defined in \Cref{fig:imp-sem}, and the fact that CHCs as defined in \Cref{def:semantics} contain a 
    premise
    of the form $y=f(x, y_1,\ldots,y_n)$. 
    In this example, we assume that all integer variables are ordered with respect to the numerical order $\leq$, and pairs of integers are ordered with respect to their pairwise integer order---i.e., $(x,y)\preceq (x',y')$ iff $x\leq x'$ and $y\leq y'$.    
    We now discuss some of the rules from \Cref{fig:imp-sem}, and whether they have a premise that uses a monotonic function.

    The function $f_{\texttt{Add}}((x,y), v_1,v_2)=v_1+v_2$, which defines the semantics of the Add rule in~\Cref{fig:imp-sem}, is monotonically increasing in the last two arguments, whereas the function $f_{\textit{Sub}}((x,y), v_1,v_2)=v_1-v_2$ is monotonically increasing in its second argument, and monotonically decreasing in its third argument. 
    Both of these functions are monotonically increasing \textit{and} decreasing in their first input argument, $(x,y)$, which only happens when the function output does not depend on that argument.

    The function {$f_{\textit{ITE}}((x,y), b, s_1, s_2) = (\textrm{if}~b~\textrm{then}~s_1~\textrm{else}~s_2)$} that defines the semantics of an if-then-else operator is monotonically increasing with respect to its first, third, and fourth arguments, but not monotonic with respect to its second argument $b$ (assuming the order $\textit{false}\preceq \textit{true}$).
\end{example}

Monotonicity is the key property that enables our interval-based pruning approach, as formulated in the following theorem.
If the function appearing in the CHC defining the semantics of a production is monotonic with respect to an appropriate set of orders, the ``endpoint recipe'' provides a sound interval abstract semantics for intervals defined over the same orders under which the function is monotonic.
The following theorem tells us how to automatically abstract the semantics of a monotone function regardless of the direction of the monotonicity.

\begin{restatable}[Abstraction of Monotonic Functions]{theorem}{monotonicity}
\label{thm:mono-interval-abstraction}
    Let {$f: Y_0 \times Y_1 \times \ldots Y_n \to Y$} be a monotone function where $\preceq_i$, {$\preceq$} are orders associated with $Y_i$, $Y$, and {$\monvec\in\{\monup,\mondown\}^{n+1}$} is a vector such that $\monvec_k=\monup$ (resp. $\monvec_k = \mondown$) if $f$ is monotonically increasing (resp. decreasing) in its $k^{\textit{th}}$ argument.
    
    We denote the \emph{endpoint extension} of $f$ to be a function {$\hat f: (Y_0 \times Y_0) \times (Y_1 \times Y_1) \times \ldots (Y_n \times Y_n)\to Y\times Y$} defined as follows: $\hat{f}(\ldots, \emptyset, \ldots) = \emptyset$; if $\lowerb_\monup(l,u)=\upper_\mondown(l,u)=l$ and $\lowerb_\mondown(l,u)=\upper_\monup(l,u)=u$
    then:
    \[
    \hat{f}([l_0,u_0],\ldots,[l_n,u_n]) = [f(\lowerb_{\monvec_0}(l_0,u_0), \ldots, \lowerb_{\monvec_n}(l_n,u_n)), {f(\upper_{\monvec_0}(l_0,u_0)}, \ldots, \upper_{\monvec_n}(l_n,u_n))]
    \]
    % where $\lowerb_\monup(l,u)=\upper_\mondown(l,u)=l$ and $\lowerb_\mondown(l,u)=\upper_\monup(l,u)=u$.
    %
    Then $\hat{f}$ is a \textit{sound interval abstraction} of $f$ in the following sense: 
    {
    \[
       \forall [l_i,u_i], l_i\preceq_i x_i\preceq_i u_i,\qquad
       \hat{f}([l_0,u_0],\ldots, [l_n,u_n])=[l,u] \Rightarrow 
       l\preceq f(x_0, \ldots, x_n)\preceq  u
    \]
    }
    Furthermore, $\hat{f}$ is the \textit{most precise abstraction} for $f$ in the following sense: if {$\hat{f}([l_0,u_0],\ldots, [l_n,u_n])=[l,u]$}, there exist {$x_0^l, \ldots, x_n^l$ and $x_0^u, \ldots, x_n^u$,} such that for every $i$, {$l_i\preceq_i x_i^l\preceq_i u_i$ and $l_i\preceq_i x_i^u\preceq_i u_i$, and $f(x^l_0, \ldots, x^l_n)=l$ and $f(x^u_0, \ldots, x^u_n)=u$.}
\end{restatable}

\begin{example}[Endpoint Extension]
\label{ex:natural-interval-extension}
    Recall the functions from \Cref{ex:monotonicity}. 
    The endpoint extension of {$f_{\textit{Plus}}((x,y), v_1,v_2)=v_1+v_2$} is the function $\hat{f}_{\textit{Plus}}([L,U], [l_1, u_1], [l_2, u_2]) = [f_{\textit{Plus}}(L, l_1, l_2), f_{\textit{Plus}}(L, u_1, u_2)]=[l_1+l_2, u_1+u_2]$ (because the function is
    monotonically increasing and decreasing in $(x,y)$, 
    we can choose either $L$ or $U$ as the first argument).
    The endpoint extension of {$f_{\textit{Minus}}((x,y),v_1,v_2)=v_1-v_2$} is the function {$\hat{f}_{\textit{Minus}}([L,U], [l_1, u_1], [l_2, u_2]) = [f_{\textit{Minus}}(L, l_1, u_2), f_{\textit{Minus}}(L, u_1, l_2)]=[l_1-u_2,u_1-l_2]$} because $f$ is monotonically increasing in its second argument and monotonically decreasing  in the third argument.
    For example, for any $[L,U]$, we have $\hat{f}_{\textit{Minus}}([L,U], [6, 7], [1,2])=[6-2,7-1]=[4,6]$, which is the most-precise interval.
\end{example}

\subsection{Automatically Generating Abstract Semantics}
\label{sec:auto-gen-sem}
Now that we have connected monotone functions to interval abstractions, we have a way to generate abstract semantics for all programs in a grammar for free, as long as the functions involved are monotone! 
For all CHCs defining the concrete semantics of constructs in our grammar, if the function $f$ appearing in the CHC can be proven monotone, we can automatically generate an associated precise interval abstract semantics by defining a new CHC where the function in the premise is simply $\hat{f}$.
When the functions are not monotone, we can conservatively define the semantics of a production to return $\top$.\footnote{
  There exist tools that can synthesize abstract semantics for arbitrary operators but require human input~\cite{amurth}.
}

This observation reduces the problem of constructing the abstract semantics to that of determining the monotonicity of semantic constructs.
If we can automatically prove the monotonicity of these constructs, then we can obtain these abstract semantics in an automated fashion for free!
The following definition tells us how to extract abstract semantics from the original program semantics.
\begin{definition}[CHC Interval Abstraction]
\label{def:chc:compilation}
Consider a CHC $Op$ as in Definition~\ref{def:semantics}:
    \begin{prooftree}
        \AxiomC{$\dennon{M_1}{t_1}({x_1}) = {y_1} \; \cdots \; \dennon{M_n}{t_n}({x_n}) = {y_n} \quad \varphi_=(x, x_1, \ldots, x_n, y_1, \ldots, y_n) \land {y} = f({x}, {y_1}, \ldots, {y_n})$}
        \RightLabel{Op}
        \UnaryInfC{$\den{\rho(t_1, \ldots, t_n)}({x}) = {y}$}
    \end{prooftree}
If $f$ is monotone, let $f^\sharp$ be defined as $\hat{f}$ (as in \Cref{thm:mono-interval-abstraction}), otherwise define it as $f^\sharp(y_1,\ldots,y_n)=\top$.
We define the \textit{interval abstract semantics} CHC $Op^\sharp$ as follows (all types are lifted to intervals):
    \begin{prooftree}
        \AxiomC{$\denanon{M_1}{t_1}({x_1}) = {y_1} \; \cdots \; \denanon{M_n}{t_n}({x_n}) = {y_n} \quad \varphi_=(x, x_1, \ldots, x_n, y_1, \ldots, y_n) \land {y} = f^\sharp({x}, {y_1}, \ldots, {y_n})$}
        \RightLabel{$Op^\sharp$}
        \UnaryInfC{$\dena{\rho(t_1, \ldots, t_n)}({x}) = {y}$}
    \end{prooftree}
\end{definition} 

% If a production has multiple associated CHCs to describe its semantics, we can construct the interval abstraction defined in~\Cref{def:chc:compilation} for each CHC, and then take the join of the resulting intervals.

To use our pruning algorithm, our abstract semantics also needs to assign semantics to holes.
The following condition connects the abstract semantics of individual holes to the conditions required for an abstract semantics to be sound (\Cref{def:abs-sem}).
\begin{definition}[Sound Abstract Semantics for a Hole]
\label{def:hole:abstraction}
Consider the CHC $\textit{Hole}_N$ assigning an abstract semantics to a hole corresponding to a nonterminal $N$ of some grammar $G$:
\begin{prooftree}
        \AxiomC{$[l,u]=f^\sharp([l_1,u_1])$}
        \RightLabel{$\textit{Hole}_N$}
        \UnaryInfC{$\denanon{N}{\hole_N}([l_1,u_1]) = [l,u]$}
    \end{prooftree}
We say that $\textit{Hole}_N$ is a \textit{sound hole abstract semantics} for $\hole_N$ if the following holds:
\[
\forall [l_1,u_1], \forall l_1\preceq_I x \preceq_I u_1, \forall P \in \LL(N), \dennon{N}{P}(x) \in \denanon{N}{\hole_N}([l_1,u_1])
\]
\end{definition}
In the above definition, setting $[l,u]$ to $\top$ is always a safe way to define a sound hole abstract semantics (as done in our abstract semantics in \Cref{fig:imp-sem}\xspace\iflabeldefined{sec:app:regex}{and \Cref{fig:regex-example}}).
This choice is also taken by the tools that our framework subsumes and that we described in Section~\ref{sec:overview}).
In \Cref{sec:precise-abstractions}, we will show an algorithm for computing more precise abstractions for holes than the $\top$ ones.

We are now ready to define an interval abstract semantics that is guaranteed to be sound and therefore safe for pruning.\footnote{As mentioned early, we assume that every production is associated with exactly one CHC. When multiple CHCs are associated with a production, the abstract semantics has to take the join of the intervals computed by all the CHCs to take into account all possible semantics.}

\begin{restatable}[Soundness of Endpoint Interval Semantics]{theorem}{intervalsoundness}
\label{thm:soundness-endpoint-semantics}
Let $G = (\mathcal{N}, \Sigma, S, \delta)$ be a regular-tree grammar with set of non-terminals $\mathcal{N}=\{N_1, \ldots N_k\}$ and $(\{\dennon{N_1}{\cdot},\ldots,\dennon{N_k}{\cdot}\},\sigma)$ a semantics for $G$.

Let $(\{\denanon{N_1}{\cdot},\ldots,\denanon{N_k}{\cdot}\},\sigma^\sharp)$ be the semantics defined as follows:
\begin{itemize}
    \item for every production $p\in \delta$, then $\sigma^\sharp(p)=\{C^\sharp \mid C\in \sigma(p)\}$ (\Cref{def:chc:compilation});
    \item for every nonterminal $N\in \mathcal{N}$, then $\sigma^\sharp(N\leftarrow \hole_N)=\{\textit{Hole}_N\}$ where $\textit{Hole}_N$ is a \textit{sound hole abstract semantics} (\Cref{def:hole:abstraction}).
\end{itemize}
Then the semantics $(\{\denanon{N_1}{\cdot},\ldots,\denanon{N_k}{\cdot}\},\sigma^\sharp)$ is a \emph{sound interval abstract semantics} for $G$.
\end{restatable}

\subsection{Extending to Nearly-Monotonic Semantics}
\label{sec:nearly-monotonic}

Recall that in \Cref{ex:monotonicity}, the function $f_{ITE}((x,y), b, s_1, s_2) = (\text{if } b \text{ then } s_1 \text{ else } s_2)$ was monotonic with respect to all its arguments, except the conditional guard $b$.
%
% There are many situations that naturally arise where a programming operator is monotonic with respect to all but a small number of arguments.
%
This section shows how one can easily modify our framework to generate an abstract semantics that is more precise than $\top$ in situations where a function is not monotonic in some of its arguments.

\paragraph{Partial Hole Filling}
When enumerating children of an if-then-else production, one will have three holes to fill: a Boolean guard $b$, and the two branches.
If the synthesis algorithm deterministically filled holes left-to-right (as done in our implementation and many other synthesizers), then the Boolean guard would be filled first.
When the conditional guard is filled concretely with a value $b_c$, the function $f_{ITE}$ can be partially evaluated to obtain a new function $g_{\textit{ITE}}((x,y), s_1, s_2) = (\textrm{if}~b_c~\textrm{then}~s_1~\textrm{else}~s_2)$ that is now monotone in all its arguments!
So, even without modifying our analysis, at this point our algorithm can compute a precise abstract transformer to decide whether the partial program (where the guard has been filled) can be pruned.

\paragraph{Enumerating Finite Domains}
In the example above, because the variable $b$ can only have two values, $\text{true}$ and $\text{false}$, we can actually attempt to compute a non-trivial abstract semantics even when the Boolean guard has not yet been filled.
Specifically, we can compute an interval overapproximation of $f_{ITE}((x,y), b, s_1, s_2)$ by taking the join of the abstractions of the two partially evaluated functions $f_{ITE}((x,y), \text{true}, s_1, s_2)$ and $f_{ITE}((x,y), \text{false}, s_1, s_2)$.
More formally, $\hat{f}_{ITE}([L,U], [\text{false}, \text{true}], [l_{s_1}, u_{s_1}], [l_{s_2}, u_{s_2}]) = [f_{ITE}(L, \text{false}, l_{s_1}, l_{s_2}), f_{ITE}(U, \text{false}, u_{s_1}, u_{s_2})] \sqcup [f_{ITE}(L, \text{true}, l_{s_1}, l_{s_2}), f_{ITE}(U, \text{true}, u_{s_1}, u_{s_2})]$.
This idea can be generalized to any setting in which a variable can only assume finitely many values by simply taking the join over all finite instantiations of that variable.
We present a complete formalization in~\iflabeldefined{sec:app:enum-domain}{\Cref{sec:app:enum-domain}}{Appendix B}.
\section{Computing Precise Hole Abstractions via Grammar Flow Analysis}
\label{sec:precise-abstractions}

% \loris{I plan to collapse this sec as a subseti of previous one}

% \twr{Actually, more to the point is to reiterate that for this section, you are assuming that monotonicity holds, and thus the ``endpoint recipe'' holds.} \rahul{fixed?}

Theorem~\ref{thm:soundness-endpoint-semantics} gives us a recipe for automatically generating interval abstract semantics for productions in the original grammar.
In this section, we assume access to such an interval abstract semantics.
However, the technique as presented so far does not give us a direct way to compute precise sound hole abstractions even when the semantics is monotonic (\Cref{def:hole:abstraction}).
Assigning the interval $\top$ to the semantics of a hole always gives us a sound hole abstraction. %, but in certain settings, this abstraction can be very imprecise.
However, formulating the problem in the \semgus framework allows us to define a more precise abstraction that can be often automatically determined.

\begin{example}[A Precise Hole Abstraction]
\label{ex:precise-hole}
    Consider a setting in which someone wants to synthesize a program without subtraction, and provides a simplified version $G_{I}^+$ of the grammar $G_I$ where the productions for nonterminal $E$ are as follows (i.e., the production $E-E$ has been removed):    
    \begin{align*}        
        E ::= 0 \mid 1 \mid x \mid y \mid E+E
    \end{align*}
    We assume the semantics is identical to the one in \Cref{fig:imp-sem}.
    This grammar can only produce terms that, when evaluated on non-negative inputs for $x$ and $y$, produce numbers that are greater than or equal to one of $x$, $y$ and 0. 
    The following hole abstraction captures this idea and is more precise than the one outputting $\top=[-\infty,\infty]$:
    \begin{prooftree}
        \AxiomC{$[l,u]=[(\textrm{if}~l_x, l_y \geq 0~\textrm{then}~0~\textrm{else}~ {-\infty}), \infty]$}
        \RightLabel{$\textit{Hole}_E$}
        \UnaryInfC{{$\denanon{E}{\hole_E}([(l_x, l_y), (u_x, u_y)]) = [l,u]$}}
    \end{prooftree}
    For example, for the input $([1,3],[2,5])$ this abstract semantics will produce the interval $[0,\infty]$, which is much more precise than $[-\infty,\infty]$.
\end{example}

\subsubsection*{Solving for One Value at a Time}
While the previous example showed a very precise abstract semantics for the hole (in fact, the most precise), 
it is challenging to come up with complex expressions automatically, like ``$\textrm{if}~l_x, l_y \geq 0~\textrm{then}~0~\textrm{else}~{-\infty}$,'' and to guarantee that they are indeed precise abstractions, for all possible values in $(l_x, l_y), (u_x, u_y)$---i.e., the problem is an expression-synthesis problem~\cite{sygus}.
However, if a specific input value $c=[(1,2),(3,5)]$ is provided, the constraints posed by  \Cref{def:hole:abstraction}, become simpler.
\begin{equation}
\label{eq:constant:hole}
\forall (x,y) \in c, \forall P \in \LL(E), \dennon{E}{P}(x,y) \in \gamma(\denanon{E}{\hole_E}(c))    
\end{equation}

Intuitively, we want $\denanon{E}{\hole_E}(c)$ to overapproximate the set of outputs any program $P \in \LL(E)$  can produce on \textit{the specific input} of $c$.
The problem of finding a solution $[l,u]$ to \Cref{eq:constant:hole} can be phrased in an abstract-interpretation framework called Grammar-Flow Analysis (GFA)~\cite{gfa}.
%\citet{nay} showed that the problem of finding a solution $[l,u]$ to \Cref{eq:constant:hole} can be phrased in an abstract-interpretation framework called Grammar-Flow Analysis (GFA)~\cite{gfa}.
% 
In GFA, this constraint can be stated using the following equations ($\sqcup$ denotes the join/union of two intervals).
%
% \rahul{The equation below uses the abstract semantics we determined using monotonicity. It's actually possible to get tighter bounds, even when a production is not monotonic. For instance, if a production $N \to x^2$, instead of producing the interval $[-\infty, \infty]$ since $x^2$ is not monotonic, my tool is able to reason that it should instead be $[0, \infty]$.}
\begin{equation}
\label{eq:gfa}
\dena{\hole_E}(c) \sqsupseteq \dena{0}(c) \sqcup \dena{1}(c) \sqcup \dena{x}(c) \sqcup \dena{y}(c) \sqcup \dena{\hole_E+\hole_E}(c)
\end{equation}
Intuitively, the constraints state that the abstract semantics of all the possible programs derivable from $E$ should be included in the semantics of $\hole_E$.

Note how the constraints are defined over the abstract semantics of the productions of that nonterminal, so we can directly apply our precise intervals based on our analysis from~\Cref{sec:auto-gen-sem}.
Because we are looking for a solution to $\gamma(\denanon{N}{\hole_N}(c))$ in the form of an interval $[l,u]$, we can rewrite \Cref{eq:gfa} as follows (where all abstract semantics have been rewritten according to their semantic rules):
\begin{equation}
\label{eq:gfa:int}
[l,u] \sqsupseteq [0,0] \sqcup [1,1] \sqcup [1,3] \sqcup [2,5] \sqcup [l+l,u+u]
\end{equation}
By unrolling the definitions, the constraint in \Cref{eq:gfa:int} can be rewritten as follows:
\begin{equation}
\label{eq:gfa:fo}
\forall v, (0\leq v\leq 0) \vee (1\leq v\leq 1) \vee (1\leq v\leq 3) \vee (2 \leq v\leq 5) \vee (l+l\leq v\leq u+u) \Rightarrow (l\leq v\leq u)
\end{equation}
As a constraint on $l$ and $u$, \Cref{eq:gfa:fo} holds for $l \leq 0$ and $u = \infty$.
The \emph{tightest} solution is $l = 0$ and $u = \infty$,
and thus the answer we seek is the interval $[0, \infty]$.
% The constraint in \Cref{eq:gfa:fo} admits every solution where $l\leq 0$ and $u=\infty$.

This procedure for computing precise hole abstractions can be lifted to any synthesis problem in the \semgus format.
Specifically, for every nonterminal $N$ in the grammar $G$, we want to construct a function $\denanon{N}{\hole_N}$ that, on an input $x$, returns a precise hole abstraction in the form of an interval $[l_N, u_N]$.
Because a semantics defined using \semgus is
represented inductively, we are able to reason about each production locally, which results in the following system of constraints:
\begin{definition}[Interval GFA]
\label{def:gfa:eqns}
Given a grammar $G$ with interval abstract semantics $\dena{\cdot}$, as in~\Cref{def:chc:compilation},
\textit{interval grammar flow analysis} is the problem of computing $\denanon{N}{\hole_N}(x)$ for each nonterminal $N \in G$,
which is defined by the following system of constraints:
\begin{equation}
\label{eq:gfa:gen}
    \denanon{N}{\hole_N}(x) \sqsupseteq \bigsqcup \{ \dena{p}(x) \mid (N \to p) \in G\} \qquad \text{for all } N \in G.
\end{equation}
\end{definition}

% Note that this system of equations is \textit{recursive}, because $\dena{p}$ can include references to nonterminal $N$ as well as other nonterminals of $G$.
%
% Because the system of equations is recursive, computing a tighter solution to the abstract hole semantics of one nonterminal can tighten the solution of other nonterminals that it appears on the right-hand side of.
%
% This also means that in general, we cannot obtain solutions to individual equations independently of the others.

Our goal is to solve the GFA constraints to obtain the tightest possible intervals $\denanon{N}{\hole_N}(x) = [l_N, u_N]$.
To avoid dealing with SMT maximization objectives (note that the formulas generated by GFA have quantifiers), our algorithm rewrites~\Cref{eq:gfa:gen} as a first-order logical constraint for each nonterminal $N$, where the bounds $l_N, u_N$ are free variables to be solved for:
\begin{equation}
\label{eq:gfa:gen:fo}
    \forall v. \bigvee_{\{p \mid (N \to p) \in G\}} \big( v \in \dena{p}(x) \big) \Rightarrow v \in [l_N, u_N]
\end{equation}
If we then unroll the interval definitions as we did for~\Cref{eq:gfa} (i.e., $x \in [l,u]$ can be rewritten as $l \preceq x \land x \preceq u$),
the optimal solution can be computed through a linear search minimization (or maximization) procedure---e.g., by iteratively asking an SMT solver to find a tighter solution to our interval bounds $l_N, u_N$ than the one computed in a previous step.
\changed{
The following theorem establishes when this iterative procedure is guaranteed to terminate:
\begin{restatable}[Termination of Iterative GFA]{theorem}{gfatermination}
\label{thm:gfa:termination}
    Suppose that the intervals $[l_N, u_N]$ from \Cref{eq:gfa:gen:fo} belong to an interval domain $\mathcal{D}$ equipped with a partial order $\preceq$.
    If $\mathcal{D}$ contains no infinite descending chains (i.e., $\prec$ is well-founded), then any algorithm that iteratively solves for $l_N, u_N$ such that $[l_N, u_N] \sqsubsetneq \dena{\hole_N}_N$ will terminate in a finite number of steps.
% \twr{Don't we want $\sqsubsetneq$? (See main.tex for idiom to avoid interference with other symbol definitions.)}
    %
    % \rahul{Just a note here: technically, termination from a slightly weaker condition, that the subset of the domain bounding the interval endpoints doesn't have infinite descending chains (e.g. suppose $\dena{\hole_N}_N = [a,b]$ where $a,b \in \mathbb{Z}$. Since $a < l_N \leq u_N < b$, there are only finitely many ways to bound $l_N$ and $u_N$. So, even though $\mathbb{Z}$ has infinite descending chains, we're still guaranteed termination since the subdomain $[a,b] \subset \mathbb{Z}$ has no infinite descending chains). But this seems needlessly complicated for what we're trying to show.}
\end{restatable}
}

\changed{
Because \Cref{thm:gfa:termination} proves that there are only finitely many possible iterations when the intervals are members of a domain with no infinite descending chains, the following approach is guaranteed to terminate under this condition:
%
% \twrchanged{
we first construct the formula in \Cref{eq:gfa:gen:fo},
% }
starting with the largest possible interval, and iteratively query tighter bounds on each of the interval endpoints (e.g., the free variables $l_N, u_N$).
We repeat this process by updating the bounds of our abstract semantics $\dena{\hole_N}_N$ with our previously solved bounds $l_N, u_N$ until the constraint solver returns unsat.
One advantage of this process is that at termination, when the constraint solver returns unsat,
% \twrchanged{
the interval-hole abstract semantics is guaranteed to be the most precise solution.
%\twr{value?}.
% }
%
Another advantage is that, even if
% \twrchanged{
this process is terminated with an intermediate solution (i.e., the constraint solver has not yet returned unsat, but, e.g., iterative GFA had exceeded some timeout threshold),
% }
this tightened interval may not be most-precise, but is still a sound hole abstraction
% \twrchanged{
\cite{DBLP:journals/entcs/BarrettK10}.
% }
%
Note that solving for hole abstractions individually is not optimal,
% \twrchanged{
because a tighter hole abstraction for one nonterminal can lead to tighter solutions for
% }
other hole abstractions. 
Thus, we define our constraints over all hole abstractions simultaneously, where we aim to tighten at least one of the intervals
% \twrchanged{
on each iteration.
% }
% Note that it is possible that tightening one hole abstraction can lead to the tightening of another.
%
% Thus, we solve for solutions to tighter interval hole abstractions of all holes simultaneously. 
}
Our solver takes this approach.\footnote{
  Abstract interpretation over intervals often involves using a widening operator.
  With the method described above, no widening is needed because
\changed{
  it starts from $\top$---an over-approximation---and makes a sequence of calls to a logic solver, rather than starting from $\bot$---an under-approximation---and performing successive approximation.
}
  Yet another approach would be to use algorithms for computing precise least solutions to systems of (certain classes of) interval equations
  \cite{DBLP:journals/tcs/SuW05,DBLP:conf/birthday/GawlitzaLRSSW09}.
}

The following theorem establishes that our computed precise hole abstractions are sound according to \Cref{def:hole:abstraction}:
% We refer the reader to another paper to the work by~\citet{nay} for a complete formalization of GFA and its connection to this problem.
% The following corollary follows from their result.
\begin{restatable}[Precise Hole Abstractions]{theorem}{preciseholeabstractions}
\label{thm:gfa:hole:sound}
    For every nonterminal $N$ in a grammar $G$, the following rule is a sound hole abstraction if {$GFASol(\denanon{N}{\hole_N}(x), G)$} is a valid solution to interval grammar flow analysis for the value of $x$:
    {
    \begin{prooftree}
        \AxiomC{$[l,u]=\textit{GFASol}(\denanon{N}{\hole_N}(x), G)$}
        \RightLabel{$\textit{Hole}_N$}
        \UnaryInfC{$\denanon{N}{\hole_N}(x) = [l,u]$}
    \end{prooftree}
    }
\end{restatable}

% \loris{if we have space, we can maybe add above another example about the Regex domain and how to solve for the CSV example since that's our killer app} \rahul{I'll check space once we remove the existing comments}
%
We are able to solve for precise hole abstractions in problems beyond integer arithmetic.
For instance, we can derive precise hole abstractions for the CSV-schema  example from~\Cref{sec:csv:precise}.

In practice, implementing an abstract semantics requires one to solve {$\textit{GFASol}(\denanon{N}{\hole_N}(x), G)$} for every possible input for which the semantics of {$\denanon{N}{\hole_N}$} needs to be evaluated during our enumeration algorithm. 
We will discuss this aspect in~\Cref{sec:evaluation}.
%
% \loris{do we discuss?}\rahul{the last few sentences before sec7.1?}

\section{Synthesis of Order Relations}
\label{sec:order-synth}

We have shown that abstract semantics can be easily extracted from a concrete semantics that is monotonic with respect to some (partial) order relations over the inputs and outputs.
So far, we have assumed that these order relations are given to us.
In this section, we show how to ``choose'' a best set of orders from a set of possible orderings.
Between two sets of orders, we prefer the one for which the most productions in the grammar exhibit a  monotone semantics.
\begin{definition}[Comparison of Orders]
    \label{def:best-orders}
    Given a grammar $G$, a semantics $(\den{\cdot},\sigma)$ describing variables with types in the set $\{T_1,\ldots, T_N\}$, and a set of orders $\omega=\{\preceq_1,\ldots, \preceq_n\}$, we define $\textit{Mon}_G(\omega)$ as the set of productions in $G$ for which the semantics are monotone with respect to $\omega$.

    We say that a set of orders $\omega_1$ is \textit{better} than a set of orders $\omega_2$ if 
    $|\textit{Mon}_G(\omega_1)| \geq |Mon_G(\omega_2)|$. 
    %\rahul{I changed this from $>$ to $\geq$ to for my changes to Def 6.3}
\end{definition}
The following example illustrates this definition.
\begin{example}[Bitvectors]\label{ex:bitvector-orders}
    Consider the following grammar $G_{\textit{bv}}$ for expressions
    over bitvectors.
    \[
        B ::= x \mid \texttt{bvand}~B~B\mid \texttt{bvor}~B~B \mid \texttt{bvadd}~B~B
    \]
    Assume that all variables are associated with one type: unsigned bitvectors of size 8.
    Consider a saturating semantics of \texttt{bvadd} in which addition saturates at $2^8-1$, and the following two possible orders over bitvectors:
    \begin{itemize}
        \item $v_1 \preceq_{\textit{bw}} v_2$ if every bit in $v_1$ is less than the corresponding bit in $v_2$; in this case $\textit{Mon}_G(\{\preceq_{\textit{bw}}\})$ contains the productions for \texttt{bvand} and \texttt{bvor};
        \item $v_1 \preceq_{\textit{bvleq}} v_2$ if the integer value of $v_1$ is smaller or equal than the integer value of $v_2$; in this case $\textit{Mon}_G(\{\preceq_{\textit{bvleq}}\})$ only contains the production for \texttt{bvadd}.
    \end{itemize}
    The set of orders $\{\preceq_{\textit{bw}}\}$ is better than the set of orders $\{\preceq_{\textit{bvleq}}\}$ because it causes the semantics of two productions to be monotonic instead of just one.
\end{example}

While our framework can accommodate multiple possible orders and prune the search space using multiple abstract semantics all evaluated in parallel, this process can become expensive. 
Moreover, if too few productions are monotonic, the resulting abstract semantics will often just yield $\top$
(\Cref{def:chc:compilation}) and not be able to prune any programs.
The following definition captures the problem of synthesizing a best set of orders that maximizes monotonicity.
\begin{definition}
  \label{def:order-synthesis-problem}
    Given a grammar $G$, a semantics $(\den{\cdot},\sigma)$ describing variables with types in the set $\{T_1,\ldots, T_N\}$, and a search space of orders $\Omega$, the \textit{order-synthesis problem} is to find a best set of orders $\omega=\{\preceq_1,\ldots, \preceq_n\}$ such that every order $\preceq_i\in \Omega$.
    Here, a \textit{best set of orders} is a set of orders $\omega$ such that $\omega$ is better than any set of orders in $\Omega$, i.e. $\forall \omega' \in \Omega. |\textit{Mon}_G(\omega)| \geq |Mon_G(\omega')|$.
    (Note that there can be multiple best orders.)
% \twr{
%     This definition has not said what ``best'' means. \Cref{ex:bitvector-orders} says that a ``better'' order (singular) is one that covers a greater number of productions, but we haven't said what a ``best set of orders'' (plural) means.
%     Why isn't $\omega$ the set of all orders in $\Omega$ that are monotonic, in which case there would be a single best $\omega$?
% } \rahul{we define a best set of orders in \Cref{def:best-orders}, right? the set just represents the orders for each nonterminal}
% \twr{No. \Cref{def:best-orders} defined ``better than;'' it does not define ``best.''  Just because you have a better-than relation does not mean that you have a best element.} \rahul{Ahh, ok fixed}
    
\end{definition}

Our implementation only considers sets of orders $\Omega$ that are finite (and typically small), and solves the order-synthesis problem by enumerating all orders in $\Omega$, computing how many productions are monotonic for each order, and returning a best one.
We focus on enumerating orders that are conjunctions of smaller ``atomic'' orders from a pre-defined base set \changed{that can easily be augmented as needed}.
These atomic orders are required to match the type of each argument (e.g., $\leq$ for integers; $\to$ for Boolean; $\preceq_{bw}$ and $\preceq_{bvleq}$ comparison for bitvectors, etc.).
This enumeration allow us to consider orders over complex semantic objects with mixed types.
\changed{
We found that initializing $\Omega$ with conjunctions over these atomic orders, while simple, performs well in practice and captures many reasonable orders for our monotonicity analysis.
}
\paragraph{A note on order optimality:}
There exist more advanced order-synthesis techniques that could improve the overall  synthesis procedure. 
Currently, we locally maximize the number of productions at each non-terminal, but do not reason about global optimizations.
It is possible that there is a stronger relation between the grammar and order, where it is more important to maximize productions closer to the starting nonterminal of the grammar, rather than the total number of productions.
Additionally, expanding the \changed{types of orders} considered in the order-synthesis problem \changed{(e.g., beyond conjunctive orders)} could also lead to improvements in generating
% \twrchanged{
the
% }
abstract semantics.
\changed{
However, continuing to perform naive enumeration after nontrivially expanding the set of orders could lead to an intractable order-synthesis problem.
}
Designing new algorithms for efficiently synthesizing best orders is an interesting direction beyond the scope of this paper.

\section{Evaluation}
\label{sec:evaluation}

We implemented our techniques in a tool called \toolname, which takes problems in the \semgus format as input~\cite{kim2021semantics} and consists of four components:
\rone A monotonicity checker, built over SMT solvers Z3 \cite{z3} and CVC5 \cite{cvc5}, that can determine what productions are monotonic with respect to a given partial order.
\rtwo An order synthesizer, which enumerates a set of partial orders and applies the monotonicity checker to find a solution to an order-synthesis problem. 
These orders are conjunctions over smaller predefined atomic orders from a defined base set.
\rthree A solver for grammar-flow analysis, which automatically constructs the system of equations from Section~\ref{sec:precise-abstractions}, and performs a fixed-point computation by generating an SMT formula and querying Z3 to iteratively tighten interval bounds. 
%\rahul{(this might be where we add gfa regex thing we discussed? (agg the symbols, take the star?)}
%
\rfour A top-down enumerative synthesizer for \semgus problems that takes the output from the previous steps, and generates the corresponding interval semantics, and uses it to prune partial programs during program synthesis.

% \twrchanged{
Components (i), (ii), and (iii) are independent of (iv), in the sense that their output is reusable by other synthesis tools.
The output of running (i), (ii), and (iii) is emitted as a JSON artifact,
% }
which describes what productions are monotonic and in what directions, as well as the computed hole abstractions on any provided input-output examples.
The artifact can be used by \emph{any} top-down enumerative solver to enable abstraction-based synthesis of \semgus problems.
In cases where we know a semantics is monotonic, but proving it is beyond the capabilities of current SMT solvers, the JSON format allows us to provide to a solver the information necessary to enable abstraction-based pruning.
Our synthesizer (component (iv)) currently supports example-based synthesis problems, and invokes component (iii) to compute the precise hole abstractions for all nonterminals and examples \textit{before} starting the enumeration (instead of dynamically calling them when evaluating the semantics and caching solutions, which can be be prohibitively expensive).

% \loris{in the eval we need to note that these tools [alpharegex and simpl] are for now still faster because domain specific and report their results. We don't need to run them since they use a diff format}

\subsection{Research Questions and Benchmarks}
In our evaluation, we consider as baseline a naive enumeration algorithm that does not perform any pruning---this baseline is the same one used to evaluate the \semgus unrealizability prover Messy~\cite{kim2021semantics}.
We do not compare against Messy because Messy can check whether a \semgus problem can be solved, but it cannot compute a solution.
Our limited set of baselines is due to the fact that \toolname is a domain-agnostic synthesizer, and we are unaware of any other synthesizers that support programs in the \semgus format or that can run all of our benchmarks.
%
% %
% \rahul{The original Semgus paper had a table showing what problems were (un/)solvable with certain solvers. Such a table be useful here too to show our generality compared to everything but MESSY}

Our experimental evaluation is designed to answer the following questions:
\begin{description}
    \item[RQ1] How effective is monotonicity-based pruning when compared to naive enumeration?
    
    \item[RQ2] How effective is our construction of precise hole abstractions from \Cref{sec:precise-abstractions}?

    \item[RQ3] How effective is \toolname at identifying interval-abstract semantics?      

    % \item[RQ4] How does \toolname compare with domain-specific synthesizers?
\end{description}
% \kjcj{Remove RQ4 and below discussion of MESSY if we don't have time.}

All experiments were run on a cluster \cite{chtc}, with each
node having an AMD EPYC 7763 64-Core Processor, of which we requested two cores and
12 GiB of RAM. 
We set a timeout value of 2000s and memory limit of 8 GiB.\footnote{Due to the nature of the computing cluster, there can be large variance
in run time between trials; however, this variance was not enough to change whether or not a problem was solvable under the time and memory constraints.}
We ran each experiment $5$ times, and report the median of these runs.

\subsubsection*{Benchmarks}
We conducted our evaluation on  \totalbenchmarks\  \semgus benchmarks from multiple domains\changed{, building on top of the \semgus Toolkit benchmark set \cite{semgus-toolkit}}.
% \twr{I don't understand where 265 comes from.
% Below you describe categories with $84 + 10 + 100 + 238 +5 = 437 \neq 265$.  What is going on?} \rahul{numbers should now be fixed: 82 + 10 + 100 + 238 + 5 = 435} \kjcj{This matches my independent calculation. But we only report synthesis on 440 - call that out here? Other 5 are from Mell et al. and only for RQ3.} \rahul{I mention it below, but if you think it would be better to do so here}
% 
% We specifically aimed at identifying benchmarks from the domain for which we knew, in addition to benchmarks that are not supported by any domain-specific synthesizer and exemplify the generality of our approach.

The first benchmark category includes 82 regular-expression synthesis problems encoded in \semgus (many of these benchmarks are part of the public \semgus benchmark set).
There are two ways to encode the semantics of regular expressions that each have pros and cons: \rone a \textit{shallow} semantics that maps a program to the corresponding term in the SMT theory of regular expressions, \rtwo a deep semantics that maps a program $r$ and a string $s$ to a Boolean \textit{matrix} that tells us what substrings of $s$ the expression $r$ accepts (see \iflabeldefined{sec:app:regex}{\Cref{sec:app:regex}}{Appendix A}).
The benchmark problems include the  AlphaRegex benchmarks~\cite{lee2016alpharegex} {(25 shallow, {25} matrix)}, regular expressions with complex operators such as negation and character classes ({1 shallow, 14 matrix}), and CSV-format synthesis problems like in \Cref{sec:csv:precise} ({15 shallow, 2 matrix}) from the CSV schema language \cite{csv-schema}.
The shallow semantics is fast to execute using corresponding regex libraries, but it has limited support in constraint solvers (e.g., no solver for the theory for of quantified formulas).
For this reason, for this semantics we must provide the JSON artifact of what productions are monotonic manually (note that effectively we only need to do so once, because all benchmarks share the same operators).
The GFA intervals for these shallow semantics can be computed by considering only intervals of the form $[\emptyset, S^*]$, where $S$ is the set of characters appearing in each grammar production.
The matrix semantics is slower to execute, but enables constraint solving.

The second benchmark category consists of 10 problems over imperative programming languages with semantics over integer variables.
% 
%10 of t
These problems are simple imperative problems that were taken from the \semgus imperative benchmark dataset.
These imperative benchmarks mostly perform loop-free operations (e.g., swap) with three benchmarks involving loops.
%
%The other 10 problems implement the task of synthesizing a score card---i.e., a sequence of if-then operations that each increments an integer by a constant---used to build a binary classifier---i.e., if the score exceeds a given amount.
%
%The thresholds and increments are randomly generated numbers with an increasing range between each benchmark, and the semantics and constraints were generated accordingly \loris{how?}.

The third benchmark category
consists of {100} problems involving
loop-free programs over bitvector variables. 
The {100} bitvector benchmarks are actually 25 concrete synthesis problems by \citet{brahma} expressed with different semantics (note that the ability to modify the program semantics is one of the key features of \semgus).
The four semantics are: \rone the traditional bitvector semantics encoded in SMTLib, \rtwo a \textit{saturated semantics}, where in the case of possible underflowing or overflowing, the result will instead take on the minimum or maximum value, respectively \cite{saturatedbv}, \rthree a semantics that introduces intermediate variables for all sub-expressions being synthesized (similar to three-address code in compilers), and \rfour a combination of the second and third semantics.
Unlike the other categories, the synthesis constraints are specified as logical formulas, requiring our enumeration algorithm to perform a Counterexample Guided Inductive Synthesis (CEGIS) loop that restarts the algorithm with new examples at each iteration.
Because of this dynamic aspect of the algorithm, we do not run the precise hole abstraction from~\Cref{sec:precise-abstractions} on this set of benchmarks; it would require solving SMT constraints while performing enumeration. 
%\rahul{what if we provided a set of I/O examples in addition to the logical constraint to run GFA here?}
%\loris{we do this for charlie paper, good idea would it help though? And worried about time}

The fourth category of benchmarks contains 238 problems where the goal is to synthesize Boolean formulas with restricted syntaxes---e.g., cubes (84 benchmarks), CNF (77 benchmarks), and DNF (77 benchmarks). 
Because a general-purpose solver cannot compete with specialized state-of-the-art synthesizers for Boolean formulas, we created a simple dataset of Boolean-function synthesis problems by randomly generating formulas of the different syntax styles with varying lengths (3-11 for cube, 2-12 for CNF/DNF) and number of variables (4-15 for cube, 4-10 for CNF/DNF).

The fifth category of benchmarks contains {5} problems adapted from \citet{video-trajectories}, a new approach for synthesizing data-classification programs over a quantitative objective function---e.g., accuracy or $F_1$ score---that manually exploits monotonicity.
%
% \loris{5 benchmarks 2 DSL?}
% \twr{clarify?  Is a "problem" a "benchmark"?
% Is it $5 \times 2$ because of the two DSLs?}
These benchmarks were created using the two DSLs presented by \citeauthor{video-trajectories} that use folding/map operators and Kleene algebra with tests, both to synthesize video trajectories.
%
% \rahulchanged{
One of the benchmarks is an encoding of their simple motivating example. The semantics of each DSL is encoded twice to create four other benchmarks with differing input sizes (because \semgus requires one to define the number of inputs).
% }
%
In contrast to the other categories, we exclusively use these benchmarks to evaluate RQ3 (i.e., whether our tool can synthesize abstract semantics), because traditional top-down enumeration is not sufficient to solve synthesis problems with quantitative objectives.
These five benchmarks are not included in our total of \totalbenchmarks.
%
% The focus for these benchmarks is to demonstrate the ability to automatically infer abstract semantics in a completely new domain.
%
% We picked these benchmarks in particular as \citet{video-trajectories} also focus on generating abstract semantics from the monotonicity of their operators.

\subsection{Effectiveness of Monotonicity-based Pruning}

% \paragraph{Quantitative Analysis}
We evaluated the effectiveness of \toolname on our benchmarks when an abstract semantics is provided;
we also measured the time taken to compute the abstract semantics (see \Cref{sec:eval:cost:sem}).
We present data for baseline top-down enumeration (\baseline),
 top-down enumeration with interval-based pruning without precise hole abstractions (\toolnamenogfa), and
 top-down enumeration with interval-based pruning with precise hole abstractions (\toolnamegfa).
In what follows, we use \toolname to denote the virtual best version of our tool that runs \toolnamenogfa and \toolnamegfa in parallel and reports the result of the first terminating instance, and \virtualbest to denote the virtual best version solver that runs \toolnamenogfa, \toolnamegfa, and \baseline in parallel.

\changed{
\Cref{tab:performance_bins} provides an overview comparing the performance between \toolname and \baseline.
\toolname ({241/\totalbenchmarks} solved) can solve {24} more benchmarks than \baseline ({217/\totalbenchmarks} solved). 
The benchmarks that \toolname can solve but \baseline cannot, fall into the following categories: regular-expressions/CSV (10), imperative (2), bit vectors (6), and Boolean (6).
% \twr{Follow the same order as they are listed in the table.}
% \twrchanged{
% bit vectors (6), Boolean (6), regular-expressions/CSV
% }
% (10), and imperative (2).
% \twr{Why are you not breaking out CSV separately?} \rahul{these are the 4 major categories described in the previous section.}
%
\toolname also significantly outperforms \baseline on 58 additional benchmarks that both could solve.
These 82 benchmarks where \toolname markedly outperforms \baseline are typically the larger and more complex benchmarks that were solvable.
On the other hand, there are 48 benchmarks where \toolname is considerably slower than \baseline.
There were no benchmarks that only \baseline could solve that \toolname could not.
The fact that \toolname is not always faster is due to a known issue in program synthesis: computing an abstract semantics and checking whether \emph{every} partial program can be pruned can be more expensive than simply exploring the search space, especially if few programs are pruned \cite{absynthe}.
}

% \begin{table}
%     \centering
%     \begin{tabular}{c|c|ccc|c}
%         \toprule[.1em]
%          Domain & \toolname only & >15\% faster & +/-15\% & >15\% slower & \baseline only \\
%         \midrule[.1em]
%         Regex Matrix & 1 & 5 & 11 & 12 & - \\
%         Regex Shallow & 2 & 7 & 6 & 5 & - \\
%         CSV & 7 & 5 & 3 & - & - \\
%         Imperative & 2 & 4 & 4 & - & - \\
%         Bitvector & 6 & 6 & 28 & 4 & - \\
%         Boolean & 6 & 31 & 60 & 27 & - \\
%         \midrule[.1em]
%         Total & 24 & 58 & 112 & 48 & - \\
%         \bottomrule[.1em]
%     \end{tabular}
%     \caption{Benchmark performance of \toolname over \baseline, broken out by benchmark category}
%     \label{tab:performance_bins}
% \end{table}

\begin{table}
    \centering
    \changed{
    \begin{tabular}{c|c|ccc|c}
        \toprule[.1em]
         \multicolumn{1}{c|}{Domain} & \toolname only & >15\% faster & +/-15\% & >15\% slower & \baseline only \\
        \midrule[.1em]
        \multirow{3}{*}{
          \begin{tabular}{r|l}
             \multirow{3}{*}{Regex}  & Matrix \\
                                     & Shallow \\
                                     & CSV
          \end{tabular}
        }
                      & 1 & 5 & 11 & 12 & - \\
                      & 2 & 7 & 6 & 5 & - \\
                      & 7 & 5 & 3 & - & - \\
        Imperative & 2 & 4 & 4 & - & - \\
        Bitvector & 6 & 6 & 28 & 4 & - \\
        Boolean & 6 & 31 & 60 & 27 & - \\
        \midrule[.1em]
        Total & 24 & 58 & 112 & 48 & - \\
        \bottomrule[.1em]      
        \bottomrule[.1em]
    \end{tabular}
    }
    % \captionsetup{labelfont={color=blue},font={color=blue}}
    \caption{Benchmark performance of \toolname over \baseline, broken out by benchmark category}
    \label{tab:performance_bins}
\end{table}

Figure~\ref{fig:overall-time} shows a cactus plot of the cumulative number of benchmarks solved by each configuration after a certain amount of time.
\toolname (and both \toolnamenogfa and \toolnamegfa) can solve more benchmarks than \baseline in cumulatively less time.
Similar trends are observed for memory usage in~\Cref{fig:overall-memory}.
% 
% For the benchmarks that both \toolname and \baseline can solve, there is no clear winner in terms of running time: \baseline is faster for {124} benchmarks, \toolname is faster for {88} benchmarks, and 5 are identical. But, {\toolname is {1.2x} faster than \baseline} if considering a geometric mean.
% 
% The fact that \toolname is not always faster is due to a known issue in program synthesis: computing an abstract semantics and checking whether \emph{every} partial program can be pruned can be more expensive than simply exploring the search space, especially if few programs are pruned \rahulchanged{\cite{absynthe}}.
%
Figure~\ref{fig:overall-concrete} shows the number of concrete programs enumerated for each benchmark before  returning a solution, 
which reveals how many programs \toolname can prune: for {75\%} of the benchmarks, \toolname explores fewer than {50\%} of the programs explored by \baseline.
% This plot shows pruning at work; there is a distinct number (what number?) of
The trends shown in Figure~\ref{fig:overall-concrete} hint that this difference would grow if we were to consider longer timeouts.

The benchmarks that none of the tools can solve
fall into the following categories: regular-expressions/CSV {(19)}, bit vectors {(56)}, and Boolean {(114)}.
% \twrchanged{
% bit vectors {(56)}, Boolean {(114)}, and regular expressions/CSV {(19)}.
% }
% \twr{Correct? Why are you not breaking out CSV separately?}
%
All 10 imperative benchmarks could be solved.
The vast majority of unsolved benchmarks require finding large solutions that are far beyond the search space explored in the given time: 
\changed{
for instance, the Boolean benchmarks contain up to twelve clauses or 15 variables, and the bitvector benchmarks contain up to 15 imperative variables or a solution AST of size up to 59.
Moreover, because there are four instances of each bitvector problem with different semantics, and the four instances are of similar difficulty, \toolname is more likely to solve all four or none of the instances for a particular problem.
}
% Two bit vector benchmarks fail during CEGIS due to the underlying SMT solver returning unknown, and six additional bit vector benchmarks fail due to using SMT-LIB features not supported by
% the synthesizer.
%
Figure~\ref{fig:unsolved-max-size} shows how deep in the search space (i.e., what program sizes) each tool got for the benchmarks that each tool could run but not solve.\footnote{Note that the $x$-axis of Figure~\ref{fig:unsolved-max-size} represents all benchmarks that timed out for some tool, and thus is a different set than the $x$-axes of Figures~\ref{fig:overall-time}, \ref{fig:overall-memory}, and \ref{fig:overall-concrete}.}
% }
% 
In general, \toolname enumerated much larger programs than \baseline (e.g., size 40 vs.\ 27) before timing out, thus showing its ability to reach deeper into the search space.

% \paragraph{Qualitative Analysis}

\textbf{To answer RQ1:} The monotonicity-based pruning approach allows \toolname to solve more benchmarks than the enumeration \baseline ({241 vs.\ 217}), with \toolnamenogfa solving 236 and \toolnamegfa all 241. 
%\rahul{Are these 24 that only we can solve over the baseline just bc of mono, or also bc of gfa?}\kjcj{This includes mono + GFA.}
% 
For the benchmarks that both \toolname and \baseline can solve, \toolname shows modest improvements in the time taken by \toolname to solve some of the benchmarks, as well as less memory being required.
Even when the benchmarks do not finish within the time limit or memory limit, \toolname reaches larger program sizes in the search space than those explored by \baseline.
%
% \kjcj{Not sure about including the below.}
% \toolname can automatically prove monotonicity for a large variety of domains, which enables speedups in these domains.
%
% Similarly, our method easily generalizes to different types of problems in these domains.
% %

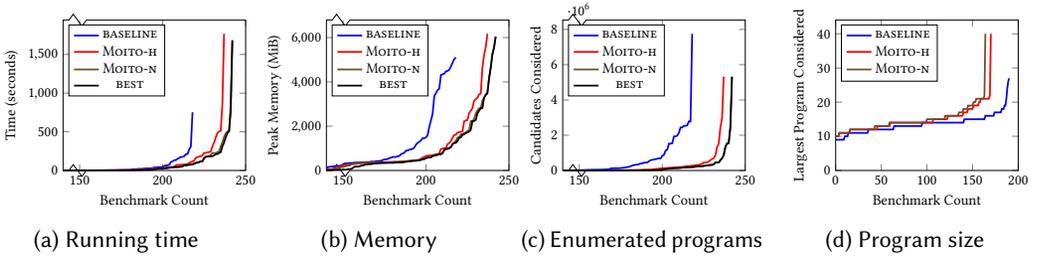
\begin{figure}
\setlength{\belowcaptionskip}{-10pt}
\hfill
\begin{subfigure}[b]{0.23\textwidth}
    \centering
     \scalebox{0.55}{
    \begin{tikzpicture}
    \begin{axis}[        
        ylabel={Time (seconds)},
        xlabel={Benchmark Count},
        xmin=140,
        xmax=250,
        ymin=0,
        width=6cm,
        legend pos=north west,
        no markers,
        axis x discontinuity=crunch
        ]
       \addplot+[sharp plot,mark=none,very thick] table [x index=0,y index=1] {figs/eval/data/cactus-time.dat};
       \addplot+[sharp plot,mark=none,very thick] table [x index=0,y index=3] {figs/eval/data/cactus-time.dat};
       \addplot+[sharp plot,mark=none,very thick] table [x index=0,y index=2] {figs/eval/data/cactus-time.dat};
       \addplot+[sharp plot,mark=none,very thick] table [x index=0,y index=4] {figs/eval/data/cactus-time.dat};
       \legend{\baseline,\toolnamenogfa,\toolnamegfa,\virtualbest}
    \end{axis}
\end{tikzpicture}
    }
    \caption{Running time}
    \label{fig:overall-time}
\end{subfigure}
\hfill
\begin{subfigure}[b]{0.23\textwidth}
    \centering
    \scalebox{0.55}{
    \begin{tikzpicture}
    \begin{axis}[
        ylabel={Peak Memory (MiB)},
        xlabel={Benchmark Count},
        xmin=140,
        xmax=250,  
        ymin=0,
        width=6cm,
        legend pos=north west,
        no markers,
        axis x discontinuity=crunch]
       \addplot+[sharp plot,mark=none,very thick] table [x index=0,y index=1] {figs/eval/data/cactus-memory.dat};
       \addplot+[sharp plot,mark=none,very thick] table [x index=0,y index=3] {figs/eval/data/cactus-memory.dat};
       \addplot+[sharp plot,mark=none,very thick] table [x index=0,y index=2] {figs/eval/data/cactus-memory.dat};
       \addplot+[sharp plot,mark=none,very thick] table [x index=0,y index=4] {figs/eval/data/cactus-memory.dat};
       \legend{\baseline,\toolnamenogfa,\toolnamegfa,\virtualbest}
    \end{axis}
\end{tikzpicture}
    }
    \caption{Memory}
    \label{fig:overall-memory}
\end{subfigure}
\hfill
\begin{subfigure}[b]{0.23\textwidth}
    \centering
    \scalebox{0.55}{
    \begin{tikzpicture}
    \begin{axis}[
        ylabel={Candidates Considered},
        xlabel={Benchmark Count},
        xmin=140,
        xmax=250,
        ymin=0,
        width=6cm,
        legend pos=north west,
        no markers,
        axis x discontinuity=crunch]
       \addplot+[sharp plot,mark=none,very thick] table [x index=0,y index=1] {figs/eval/data/cactus-concrete.dat};
       \addplot+[sharp plot,mark=none,very thick] table [x index=0,y index=3] {figs/eval/data/cactus-concrete.dat};
       \addplot+[sharp plot,mark=none,very thick] table [x index=0,y index=2] {figs/eval/data/cactus-concrete.dat};
       \addplot+[sharp plot,mark=none,very thick] table [x index=0,y index=4] {figs/eval/data/cactus-concrete.dat};
       \legend{\baseline,\toolnamenogfa,\toolnamegfa,\virtualbest}
    \end{axis}
\end{tikzpicture}
    }
    \caption{Enumerated programs}
    \label{fig:overall-concrete}
\end{subfigure}
\hfill
\begin{subfigure}[b]{0.23\textwidth}
    \centering
    \scalebox{0.55}{
    \begin{tikzpicture}
    \begin{axis}[
        ylabel={Largest Program Considered},
        xlabel={Benchmark Count},
        xmin=0,
        xmax=200,
        ymin=0,
        width=6cm,
        legend pos=north west,
        no markers]
       \addplot+[sharp plot,mark=none,very thick] table [x index=0,y index=1] {figs/eval/data/cactus-max-size.dat};
       \addplot+[sharp plot,mark=none,very thick] table [x index=0,y index=3] {figs/eval/data/cactus-max-size.dat};
       \addplot+[sharp plot,mark=none,very thick] table [x index=0,y index=2] {figs/eval/data/cactus-max-size.dat};
       \legend{\baseline,\toolnamenogfa,\toolnamegfa}
    \end{axis}
\end{tikzpicture}
    }
    \caption{Program size}
    \label{fig:unsolved-max-size}
\end{subfigure}
\caption{
The first three plots compare \toolnamenogfa, \toolnamegfa, and \baseline across time, memory, and number of enumerated programs (we start the $x$-axis at {150} to better illustrate the interesting behavior---i.e., we do not show the behavior on the {150} easiest benchmarks). 
The last plot shows the maximum size reached before timing out on problems that could not be solved. Note that a line that is lower and to the right represents better performance in the first three plots, and higher and to the left in the final plot.}
\end{figure}

\changed{
\paragraph{A note on comparisons to state-of-the-art solvers:}
\toolname can solve {19/25} AlphaRegex~\cite{lee2016alpharegex} benchmarks within our timeout of 2000s.
% \twr{Say here what the timeout threshold was.}
(when using the {shallow} semantics), whereas AlphaRegex reports solving 25/25 in less than a minute each.
AlphaRegex implements many other domain-specific optimizations (e.g., a simple form of equality saturation) that our tool cannot automate because the input problem is given as an arbitrary \semgus file.
For Simpl~\cite{so2017synthesizing}, a direct comparison is more difficult, because their benchmark set focuses on loops. 
\toolname cannot prove loops monotonic, and so our technique would not produce pruning benefits on these benchmarks.
On the other hand, many of our non-Boolean benchmarks (97/192) are beyond the reach of existing customized synthesizers. 
Specifically, 5/39 matrix regular expressions are not expressible in AlphaRegex (due to negation operators). None of the CSV benchmarks (17/17) are supported by AlphaRegex (due to character classes beyond “0”/“1” and CSV-format-specific grammars).
75/100 of the bitvector benchmarks are not solvable using Brahma~\cite{brahma}, because the benchmarks use an alternative semantics (imperative and/or saturating semantics). 
%
% \twrchanged{
% For all of the 238 Boolean benchmarks in the \semgus format (which each use a set of I/O examples as the input specification), we are not aware of any customized solver that we can compare \toolname against.
% }
% \twr{Correct?}
%
This generality is an advantage of a parameterized framework like SemGuS, where users can apply a solver that supports all instantiations of the framework, rather than a dozen domain-specific tools with their own restrictive DSLs.
%
% Furthermore, while these benchmarks are not supported by their original tools, \toolname is general enough to generate interval abstractions for these problems completely for free, without any additional human engineering effort.
}

\subsection{Effectiveness of Precise Hole Abstractions}
\label{sec:eval:gfa}
% \rahul{We observe speedups on certain benchmarks (CSV, hopefully bitvectors), little effect on others}
%\note{Things to measure:
%\begin{itemize}
%    \item Time and \#ASTs between baseline enumeration, with monotonicity no gfa, and with monotonicity + gfa
%    \item Time to perform GFA before synthesis (and average time per example)
%    \item Number of iterations to converge (and average num iterations)
%\end{itemize}
%}
\Cref{fig:gfa-time-compare} compares \toolnamegfa and \toolnamenogfa using a scatter plot. Both variants seem to be beneficial
\begin{wrapfigure}{r}{0.33\textwidth}
\vspace{-3mm}
    \centering
    \scalebox{0.75}{
    \begin{tikzpicture}
    \tikzset{mark options={mark size=1, opacity=0.5}}
    \begin{axis}[
        ylabel={\toolnamegfa Runtime (s)},
        xlabel={\toolnamenogfa Runtime (s)},
        ymode=log,
        xmode=log,
        xtick={1,10,100}, ytick={1,10,100},
        width=6cm]
       \addplot+[only marks] table [x index=1,y index=2] {figs/eval/data/gfa-compare.dat};
       \addplot+[mark=none,solid, color=black] coordinates {(0.1,0.1) (2000,2000)};
    \end{axis}
\end{tikzpicture}
    }
    \vspace{-2mm}
    \caption{\toolnamegfa vs \toolnamenogfa (log,log)}
    \label{fig:gfa-time-compare}
    \vspace{-4mm}
\end{wrapfigure}
in different settings, although on average, \toolnamegfa is 6\% faster than \toolnamenogfa (geomean, variance 1.06), {and \toolnamenogfa explores
twice as many programs as \toolnamegfa.} Additionally, there
were 5 CSV benchmarks that were only solved by \toolnamegfa.
We conjecture that \toolnamegfa is sometimes slower than \toolnamenogfa because
while \toolnamegfa can compute more precise hole abstractions, propagating the semantics of intervals different than $\top$ is generally more expensive.
Therefore, in cases where the increased precision does not prune more programs, \toolnamegfa is slower.

As expected, \toolnamegfa performs much better than \toolnamenogfa on benchmarks where it can prune many more programs\changed{, and about the same where it cannot.}
\changed{
This bimodality is evident in Figure~\ref{fig:gfa-time-compare}; most benchmarks are on the 1x (no improvement) diagonal line, but there is a subset of benchmarks below the 1x diagonal line, showing substantial improvement. A frequency analysis shows the largest cluster of benchmarks within 0.95x - 1.05x improvement (172 of 236), with a second cluster of 18 benchmarks above 1.20x improvement (max of 6x improvement).
Similar trends are seen for the number of concrete candidate programs considered for each benchmark: most benchmarks (154 of 236) do not check any fewer programs with \toolnamegfa, with a long tail of benchmarks showing improvement (10 benchmarks check over 100x fewer programs).
} 
%\kjcj{If room (or appendix), we could throw in a histogram plot.}
% 
For example, for the {10/17} benchmarks in the CSV category solved by both tools, \toolnamegfa explores on average only 0.1\% of the programs explored by \toolnamenogfa, and is on average 2.6x faster.
% In fact, \rahulchanged{on} many of the CSV benchmarks ({7/10})\rahulchanged{, \toolnamegfa only had to evaluate} one or two concrete programs, compared to between 52 and 16288 with \toolnamenogfa. 

Benchmarks in other categories also took advantage of \toolnamegfa.
% \twrchanged{
For example, for the imperative benchmark ``max3-impv,'' (which computes the maximum of 3 integers), \toolnamegfa \rone proved that there were non-terminals that returned only values in the interval $[0, \infty]$, 
\rtwo ran 24\% faster, and
\rthree checked 47\% fewer concrete programs.
% }
For other cases where \toolnamegfa could not compute better intervals than $\top$, both \toolnamegfa and \toolnamenogfa checked the same number of concrete programs and had similar running times,
% \twrchanged{
as expected.
% }
% 
% One is \textit{swap2-impv}, an imperative benchmark to synthesize a program that swaps two numbers.
% In this case, we found that all output variables would only
% take on positive integers, so any partial programs that produced a negative lower bound
% could be discarded. Another is \textit{csv\_03-shallow}, a regular-expression benchmark.
% This benchmark has a grammar that is split into terms that match letters and numbers,
% so we could construct intervals for, e.g., numbers, with \texttt{(re.none)} as the bottom 
% and \texttt{(re.* (re.range "0" "9"))} as the top.

% In this section, we evaluate the effectiveness of generating precise hole abstractions in \toolname. 
%
% Figures~\ref{fig:gfa-time-compare} and \ref{fig:gfa-ast-compare} show the percentage difference in runtime and
% concrete candidate programs checked with and without precise hole abstractions.

\textbf{To answer RQ2:} %\toolnamegfa and \toolnamenogfa have incomparable performances and are both beneficial.
\toolnamegfa is very beneficial for instances when the structure of the grammar restricts the possible values of certain nonterminals, and therefore the more precise hole abstractions can prune many programs.

\subsection{Effectiveness of Computing Abstract Semantics}
\label{sec:eval:cost:sem}
% Table~\ref{tab:pipeline-timing} reports the time taken by \toolname to compute the abstract semantics for all operators in the grammar and to compute precise hole abstractions for all the input examples.
% 
% \begin{table}
% \caption{Timing \twr{Time?} to Compute Interval Abstract Semantics}
%     \label{tab:pipeline-timing}
%     \centering
%     \begin{tabular}{c|cc|cc}
% \toprule[.1em]    
%         \ & \multicolumn{2}{|c|}{Interval Abstract Semantics} & \multicolumn{2}{|c|}{Precise Hole Abstractions}\\
%         Domain & Time Range & Average & Total Time & Time per Example \\
%         \midrule[.1em]
%         Boolean          & [0.45, 0.71]  & 0.59  & 11 & 0.13 \\
%         Bit Vectors      & [0.99, 130] & 11  & \multicolumn{2}{l}{Unsupported}\\
%         Matrix Regex     & [5.4, 1000+]  & 150 & 11 & 1.2 \\
%         Other Imperative & [1.2, 1.4]  & 1.3 & 0.82 & 0.91 \\
%   \bottomrule[.1em]
% \bottomrule[.1em]      
%     \end{tabular}    
% \end{table}

\toolname can compute the interval abstract semantics for {373}/430\ 
 of the original benchmarks.
\toolname can also prove monotonicity on all five benchmarks adapted from \citet{video-trajectories}, where \toolname was able to automatically generate the abstract semantics that were manually defined by \citeauthor{video-trajectories}.
% \twr{This sentence makes no sense because you said that the 5 benchmarks from Mell are not counted in the 430.  Reword it.}

Computing the abstract semantics timed out on 12 of the imperative benchmarks over bitvectors and 4 regular-expression with matrix-semantics benchmarks. 
These benchmarks involved semantics in which functions take 10-100 variables as input, thus causing the order-synthesis algorithm to consider many possible order combinations.
As we mentioned, there is currently no solver that supports the quantified theory of regular expressions, which is necessary for computing an abstract semantics of the {41} regular-expression benchmarks with shallow semantics.
% \twr{
% Again, I don't understand your counts of the numbers of benchmarks.  You talk about computing the interval semantics for 249 out of 265 benchmarks.
% Thus, you couldn't do it for 16 benchmarks, which are the $12 + 4$.  So what are the 29 regular-expression benchmarks with shallow semantics?
% They can't be part of the 265, so you have 265 + 29 benchmarks?  How am I supposed to understand that number compared to the $84 + 10 + 100 + 238 +5 = 437$ benchmarks discussed earlier.
% }

The time to compute the abstract semantics varied across domains.
All the Boolean benchmarks could terminate in less than a second (avg 0.57s) and the imperative benchmarks took around 0.8 to 5 seconds each (avg 2.1s).
The variance was larger for regular-expressions and bitvector benchmarks, and we show a detailed analysis of these categories in Figure~\ref{fig:pipeline-time-compare}.
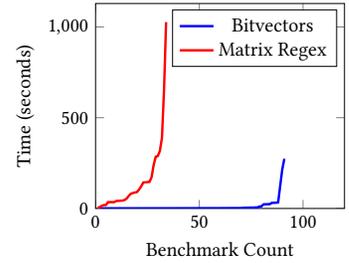
\begin{wrapfigure}{r}{0.33\textwidth}
\vspace{-3mm}
    \centering
    \scalebox{0.75}{
    \begin{tikzpicture}
    \begin{axis}[        
        ylabel={Time (seconds)},
        xlabel={Benchmark Count},
        xmin=0,
        xmax=120,
        ymin=0,
        width=6cm,
        legend pos=north east,
        no markers]
       \addplot+[sharp plot,mark=none,very thick] table [x index=0,y index=1] {figs/eval/data/pipeline-cactus.dat};
       \addplot+[sharp plot,mark=none,very thick] table [x index=0,y index=2] {figs/eval/data/pipeline-cactus.dat};
       \legend{Bitvectors,Matrix Regex}
    \end{axis}
\end{tikzpicture}
    }
    \vspace{-4mm}
    \caption{Time to compute interval abstract semantics}
    \label{fig:pipeline-time-compare}
    \vspace{-4mm}
\end{wrapfigure}
On average, it took 150s to compute the semantics for regular-expression benchmarks, and the time scales with the size of the matrices used in the semantics (i.e., the length of the input examples).
Similarly, computing the semantics of bitvector benchmarks took on average 11s with the time scaling
%\twr{What is meant by ``scaling''?  Do you mean ``increasing linearly''? ``increasing roughly linearly''?
%``Scaling'' used in the way you used it is a waffling word that should not be used.
%} \rahul{exponentially?}
exponentially with the number of variables in the considered programming language, because the size of the order search space grows exponentially.

We note that in practice, an abstract semantics does not need to be recomputed every time the specification of the input problem changes, as long as the language remains the same.
% 
% Furthermore, our current strategy for synthesizing orders is fairly naive and this aspect of our framework could be improved if one were to identify better order-synthesis techniques.
% % 

While \toolname always outputs an abstract semantics, the individual CHCs are only precise abstractions if the original semantics was monotonic.
\toolname can identify orders under which the semantics is monotonic for all the productions in the grammars for regular expressions, Boolean, and
imperative programs.
In the case of {bitvectors}, \toolname can find at least one order for either the traditional or saturating semantics for ${\sim}{87}\%$ of productions.
%\kjcj{Not sure where this number came from:} \twrchanged{
%${\sim}{83}\%$ of the productions have a monotonic semantics (at the same time).
%} \rahul{it wasn't at the same time, right? @Keith do you remember what you were trying to say here?}
% 
We observed variability in these benchmarks;
in one case only 6 out of 11 productions could be proven monotonic.
An example showing why a bitvector semantics is not monotonic for all productions
was illustrated in \Cref{sec:order-synth}.

% %
% However, due to the overhead of constructing these intervals in practice \rahul{(I assume we report this data somewhere)}, we only use a single order to generate all of our intervals.

Computing precise hole abstractions takes on average {0.73}s per input example, although the time can vary across different applications (regular expressions {1.3}s, imperative {9.9}s, Boolean {0.02}s).
Although this step can be costly, we note that \toolnamegfa can sometimes {provide large performance gains on certain benchmarks that \toolnamenogfa cannot} (\Cref{sec:eval:gfa}).
{As previously mentioned, we omit the bitvector benchmarks from this analysis because they use logical specifications.}

% \loris{did you discuss mell data?}

\textbf{To answer RQ3:} \toolname can discover precise abstract semantics (and precise hole abstractions) for most benchmarks.
This result confirms that our framework can \emph{automatically} discover many of the domain-specific techniques used in existing tools, and generalizes them to new domains (e.g., Boolean formulas, bitvector programs, and DSLs for video trajectories).

\section{Related Work}
\label{sec:related-work}

\subsubsection*{Top-Down Enumeration and Pruning}
A number of papers have addressed the problem of program synthesis by applying top-down enumeration in specific domains, such as regular expressions \cite{lee2016alpharegex}, imperative programs \cite{so2017synthesizing}, SQL queries \cite{wang2017scythe}, Datalog programs \cite{si2018datalog}, and functional programs \cite{synquid}.
These tools differ from our work in two key ways.
First, our approach applies to arbitrary synthesis problems defined in the \semgus framework, whereas these tools each implement a solution to \emph{one fixed} domain-specific synthesis problem.
Because of this specificity, these tools outperform \toolname on their respective tasks, but their implementations are monolithic and tailored to such tasks. 
Second, while these tools implement hard-coded pruning strategies, \toolname automatically discovers pruning opportunities by extracting an abstract semantics for the given \semgus problem.
In summary, our work can automatically discover ways to prune in top-down enumeration for problems defined in the \semgus framework, while these tools use manually-defined pruning strategies that target specific domains.

Besides subsuming several of the pruning strategies used by AlphaRegex~\cite{lee2016alpharegex} and SIMPL~\cite{so2017synthesizing},
our interval-based framework also captures some of the pruning approaches used when synthesizing SQL queries \cite{wang2017scythe}, Datalog programs \cite{si2018datalog}, and data-processing tasks \cite{video-trajectories}.
Scythe~\cite{wang2017scythe} (indirectly) uses an interval $[T_l,T_u]$ (where $T_l$ and $T_u$ are tables) to represent what possible output tables could be the result of evaluating the completion of a partial SQL query, and uses the fact that most queries are monotonic with respect to the predicates appearing in a where-clause---i.e., a more permissive clause yields a bigger table.
\citet{si2018datalog} use a similar insight to construct an interval (akin to a version-space algebra)  over the set of Datalog programs that are consistent with a set of input-output examples.
%
% \rahulchanged{
% These intervals are constructed over meta-rules with relation symbols to be filled in, with monotonic refinement operators to construct further precise intervals.}
% \loris
% 
We do not evaluate our approach on these applications because the \semgus format currently lacks some features that are necessary to model these applications (e.g., the theory of bags for SQL and fixed-point logics for Datalog), and the sizes of the inputs used in these domains (e.g., tables) are only within the reach of domain-specific tools.

\citet{video-trajectories} use interval abstractions when productions are monotonic to guide their search for optimal synthesis, and their specific monotonicity property can be viewed as a specific instantiation of ours. 
Their work focuses on two specific domains, that of numbers and Booleans under their standard orders.
While \citeauthor{video-trajectories} had to \textit{manually} prove monotonicity \textit{a priori}, \toolname was able to \textit{automatically} infer that the semantics constructs used in their experiments are monotonic (\Cref{sec:evaluation}).
Although their work uses monotonicity in an additional way---i.e., to maximize an objective function---their tool could be another client of our automated monotonicity analysis.

\subsubsection*{Abstraction-Guided Synthesis}
Complex forms of abstraction-based pruning have been applied in many other program-synthesis tools \cite{so2017synthesizing, vechev2010, 2017wangcegar}.
However, these tools are also domain-specific and cannot tackle \semgus problems.

\textsc{Simba} \cite{simba} combines forward abstract interpretation (for soundly approximating the set of possible outputs obtainable from inputs of partial programs) with backward abstract interpretation (for approximating the set of possible inputs, starting from the outputs).
\textsc{Simba} only supports \sygus problems (i.e., expression-synthesis problems) and requires the user to manually provide highly-precise abstract semantics.
Despite this limitation, an interesting research direction is whether \textsc{Simba}'s approach can be automatically generalized to \semgus in the same way that our framework automatically generalizes interval-based pruning---i.e., we hope our work will be the first of many in this spirit. 
\textsc{Absynthe} \cite{absynthe} is a general-purpose framework for synthesis with abstraction-based pruning that allows users to manually supply abstract semantics for the language over which synthesis is being performed.
In contrast, \toolname can automatically discover a precise abstract interval semantics from the user-provided concrete semantics.

\subsubsection*{Generating Abstract Semantics} 
Amurth~\cite{amurth} synthesizes abstract semantics for arbitrary languages by asking a user to provide a grammar of possible abstract functions to choose from (i.e., our function $f^\sharp$ in \Cref{def:chc:compilation}).
While Amurth is very general, it is not fully automated and requires the user of the tool to provide a specialized grammar for each abstract domain---and in many cases for each abstract function.
(These grammars often contain complex insights on what a particular abstract function should look like.)
Instead, our work is based on monotonicity conditions under which abstract semantics can be generated \textit{automatically}.
Combining Amurth with our tool to generate abstract semantics beyond the automatically generated ones discussed in this paper is an interesting research direction, although it would require ways to identify what grammars one should provide to Amurth.

Atlas \cite{atlas} also learns abstractions for pruning in synthesis. 
%
% While Atlas and \toolname serve to generate abstract transformers for program synthesis, the goals are different. 
%
Atlas considers abstract domains consisting of linear equalities, which work better than intervals in certain settings---e.g., reasoning about string lengths---but are limited to numerical domains---e.g., they cannot capture the conjunctions of Booleans used in many of our benchmarks. 
%
% Also, the time it takes \toolname to determine abstract transformers over a new DSL is significantly smaller than that for Atlas, which has to solve second-order equations within a refinement loop. 
%
Most importantly, Atlas requires a set of training problems to synthesize an abstract domain and the transformers (i.e., the tool needs to be trained for every new domain), while \toolname is domain-agnostic, does not require a training phase, and it synthesizes an abstract semantics directly from the provided concrete semantics alone. 
%
% An interesting research direction would be to see how Atlas could be adapted to interface with \semgus solvers in conjunction with \toolname.

\subsubsection*{Other Forms of Enumeration}
% Besides top-down enumeration, there are other approaches to solving synthesis problems.
% 
Bottom-up enumeration enumerates \textit{subprograms} of increasing size derivable from each nonterminal in the grammar, and prunes the search space by only maintaining programs that are observationally inequivalent on the examples~\cite{eusolver}.
Observational equivalence can be automated for expression-synthesis problems, but not, for example, for synthesizing imperative programs.
This limitation is due to the fact that in an imperative programming language, different subprograms are evaluated on different states (i.e., the programs are stateful).
For the same reason, there is currently no way to implement bottom-up enumeration (with pruning) for \semgus problems, hence our focus on top-down enumeration.
Hybrid versions of top-down and bottom-up enumeration share the same limitation~\cite{duet}.

\subsubsection*{Symbolic Solvers}
\textsc{Messy} \cite{kim2021semantics} is currently the only published \semgus solver, and therefore the only solver that is designed to solve the same range of tasks as \toolname.
\textsc{Messy} employs a constraint-based approach for solving \semgus problems using a Constrained Horn Clause solver with the dual goals of being able to synthesize a program or to prove that the synthesis problem is unrealizable.
While \textsc{Messy} performs well for proving unrealizability of \semgus problems, it has effectively no synthesis capabilities (i.e., it cannot produce an output program when a problem is realizable); therefore we do not compare \toolname against it in our evaluation.

The baseline implementation of top-down enumeration used in our evaluation is the same as the baseline called \textsc{MessyEnum} used to evaluate \textsc{Messy}.
To the best of our knowledge, \toolname is the first enumeration technique for \semgus problems that moves beyond naive enumeration.

\section{Conclusion}
This paper presents a unified framework for determining precise interval abstract semantics that can speed up program synthesis via enumeration for problems written in the \semgus framework.
Unlike existing works on top-down enumeration, our framework is domain-agnostic (i.e., it does not know \textit{a priori} the semantics of programs appearing in the search space).

\changed{
Recall that the solvers in the initial SyGuS competition were unable to solve a majority of the original SyGuS benchmarks.
The difficulty of the benchmark set then spurred the development of second- and third-generation solvers that could solve most of the competition benchmarks \cite{sygus-comp}.
}
\semgus and the existing \semgus solvers are currently in their infancy---i.e., they have limited 
\changed{
scalability---but we hope that \toolname is the first in a series of improved \semgus solvers.
In particular,
}
our work opens the door for generalizing and automating many other domain-specific synthesis techniques, so that they can be lifted to a general framework like \semgus.
For example, the Amurth~\cite{amurth} tool for synthesizing abstract transformers or the Atlas~\cite{atlas} tool for constructing linear abstractions could be combined with our theory to specialize our work to more complex abstract domains.
In the same way domain-specific insights have caused tremendous speedups in \sygus solvers (which initially could only solve trivial problems), we are hopeful that future efforts similar to the one described in this paper will result in \semgus solvers that---despite their generality---can solve complex problems.

% \rahulchanged{
\section*{Data-Availability Statement}
The software that supports~\Cref{sec:evaluation}
is available on Zenodo~\cite{zenodo}.

\begin{acks}
This work was supported, in part,
by a gift from
\grantsponsor{00001}{Rajiv and Ritu Batra}{}, and by
\grantsponsor{00003}{NSF}{https://www.nsf.gov/}
under grants
\grantnum{00003}{CCF-2211968}
and
\grantnum{00003}{CCF-2212558}.
Any opinions, findings, and conclusions or recommendations
expressed in this publication are those of the authors,
and do not necessarily reflect the views of the sponsoring
entities.
\end{acks}

\bibliography{bibliography.bib}

\ifthenelse{\boolean{showappendices}}{
  % Code to include appendices
  \appendix
  \section{Regular Expression Case Study}
\label{sec:app:regex}
This section illustrates how our framework can be instantiated for the problem of synthesizing regular expressions.

\subsection{Grammar and Semantics}

The grammar $G_R$ in \Cref{fig:regex-grammar} defines a language of regular expressions and allows regular expressions to contain complement operations.
A typical formulation for the semantics of regular expressions maps a string to a Boolean value that denotes whether or not the string is accepted.
However, this semantics introduces nondeterminism and does not lend itself well to, e.g., efficiently checking if a synthesized regular expression accepts a string.

We instead present a deterministic semantics that, due to its efficiency, is commonly used in hardware accelerators~\cite{loris-micro} (Figure~\ref{fig:regex-sem}).
For a regular expression $r$ and string $s=a_1,\ldots, a_{k-1}$, it outputs a \textit{Boolean matrix} $A$ of size $k \times k$, such that the entry $A_{i,j}$ is set to true (represented as 1) if and only if the substring $a_i,\ldots, a_{j-1}$ is accepted by $r$ (if $i=j$ the substring is the empty string and if $i>j$ the entry $A_{i,j}$ is always 0).
The semantics of the concatenation (union) of two regular expressions is then simply the matrix multiplication (sum) of the corresponding matrices.
Using these two primitives and the identity matrix \textbf{I}, we can define the semantics of $r^*$ as the semantics of the union of up to $k$ concatenations of $r$ with itself.

Using $\textbf{1}$, the matrix in which all elements in the upper triangle---including the diagonal---are 1, and all elements in the lower triangle are 0, we can define the semantics of $\neg r$ as $\textbf{1}$ minus the semantics of $r$ (i.e., the complement of the matrix).

\begin{figure}
    % \label{fig:regex-imp}
    \begin{subfigure}{0.3\textwidth}
        \small{
            \begin{align*}
                S &::=\textit{accepts}(R) \\
                R &::= c \mid \epsilon \mid \emptyset \mid (R \mid R) \mid (R\cdot R) \mid R^* \mid \neg R
            \end{align*}
        \caption{Grammar $G_R$.}
        \label{fig:regex-grammar}
        }
    \end{subfigure}
    \hfill
    \begin{subfigure}{0.3\textwidth}
        \footnotesize{
            \begin{align*}
                \{ (1,\textit{true}), (10,\textit{true}), (111,\textit{true}), \\
                (0,\textit{false}), (00,\textit{false}), (100,\textit{false}) \}
            \end{align*}
        \caption{Examples $\examples^R_1$}
        \label{fig:regex-constraint}
        }
    \end{subfigure}
    \hfill
    \begin{subfigure}{0.3\textwidth}
        \footnotesize{
            \centering
            $(1\cdot(0\mid 1))^*\cdot 1$
            \vspace{4pt}
        \caption{Solution to $\examples^R_1$}
        \label{fig:regex-solution}
        }
    \end{subfigure}
    \vspace{14pt}

    \begin{subfigure}{0.35\textwidth}
        \tiny{
        % \begin{tikzpicture}
        % \node[text width=7cm,draw,inner sep=0.6em](mybox){
            \centering
            \begin{prooftree}
                \AxiomC{$\den{r}(s) = A_1$}
                \AxiomC{$y = (A_{0,l} == 1)$}
                \RightLabel{Accepts}
                \BinaryInfC{$\den{accepts(r)}(s) = y$}
            \end{prooftree}
            \vspace{0.5em}
            \begin{minipage}{0.3\textwidth}
                \begin{prooftree}
                    \AxiomC{$A = \mi$}
                    \RightLabel{Eps}
                    \UnaryInfC{$\dens{\epsilon}(s) = A$}
                \end{prooftree}
            \end{minipage}
            \hspace{0.8cm}
            \begin{minipage}{0.3\textwidth}
                \begin{prooftree}
                    \AxiomC{$A = \mz$}
                    \RightLabel{Empty}
                    \UnaryInfC{$\dens{\emptyset}(s) = A$}
                \end{prooftree}
            \end{minipage}
            \vspace{0.1em}
            % \begin{minipage}
                \begin{prooftree}
                    \AxiomC{$\forall 0 \leq i,j \leq l+1.~ \big(A_{i,j} = (i+1 = j \land s_i = c)\big)$}
                    \RightLabel{Character}
                    \UnaryInfC{$\dens{c}(s) = A$}
                \end{prooftree}
            % \end{minipage}
            \vspace{0.3em}
            \begin{prooftree}
                \AxiomC{$\den{r_1}(s) = A_1 \quad \den{r_2}(s) = A_2 \quad A = A_1 + A_2$}
                \RightLabel{Union}
                \UnaryInfC{$\dens{r_1 \mid r_2}(s) = A$}
            \end{prooftree}
            \vspace{0.2em}
            \begin{prooftree}
                \AxiomC{$\den{r_1}(s) = A_1 \quad \den{r_2}(s) = A_2 \quad A = A_1 \times A_2$}
                \RightLabel{Concat}
                \UnaryInfC{$\dens{r_1 \cdot r_2}(s) = A$}
            \end{prooftree}
            \vspace{0.1em}
            % \centering
            % \begin{minipage}{0.45\textwidth}
            \begin{prooftree}
                \AxiomC{$\den{r}(s) = A_1$}
                \AxiomC{$A = \mi + \sum_{i=1}^k A_1^k$}
                \RightLabel{Star}
                \BinaryInfC{$\dens{r^*}(s) = A$}
            \end{prooftree}
            % \end{minipage}
            \vspace{-0.1em}
            \vspace{\proofspacing}
            % \begin{minipage}{0.45\textwidth}
            \begin{prooftree}
                \AxiomC{$\den{r}(s) = A_1$}
                \AxiomC{$A =\mo - A_1$}
                \RightLabel{Neg}
                \BinaryInfC{$\dens{\neg r}(s) = A$}
            \end{prooftree}
            \vspace{0.4em}
            % \end{minipage}
        %     };
        %     \node[text=gray,anchor=west,fill=white,xshift=0.5em] at (mybox.north west)
        %      {\textsc{Regular Expression Semantics}};
        % \end{tikzpicture}
        \caption{CHC-based semantics}
        \label{fig:regex-sem}
    }
    \end{subfigure}
    \hfill
    \begin{subfigure}{0.6\textwidth}
        \tiny{
            \centering
            \begin{minipage}{0.4\textwidth}
                \begin{prooftree}
                    \AxiomC{$[L,U] = [\bot, \top]$}
                    \RightLabel{Hole$^\sharp_S$}
                    \UnaryInfC{$\dena{\hole_S}([\mathcal{S}]) = [L,U]$}
                \end{prooftree}
            \end{minipage}
            \hspace{0.8cm}
            \begin{minipage}{0.4\textwidth}
                \begin{prooftree}
                    \AxiomC{$[L,U] = [\emptyset, (0 \mid 1)^*]$}
                    \RightLabel{Hole$^\sharp_R$}
                    \UnaryInfC{$\dena{\hole_R}([\mathcal{S}]) = [L,U]$}
                \end{prooftree}
            \end{minipage}
            \begin{prooftree}
                \AxiomC{$\dena{r}([\mathcal{S}]) = [L', U']$}
                \AxiomC{$[L,U] = [(L_{0,l}' == 1), (U_{0,l}' == 1)]$}
                \RightLabel{Accepts$^\sharp$}
                \BinaryInfC{$\dena{\textit{accepts}(r)}([\mathcal{S}]) = [L,U]$}
            \end{prooftree}
            
            \begin{minipage}{0.35\textwidth}
                \begin{prooftree}
                    \AxiomC{$[L,U] = [\mi, \mi]$}
                    \RightLabel{Eps$^\sharp$}
                    \UnaryInfC{$\dena{\epsilon}([\mathcal{S}]) = [L,U]$}
                \end{prooftree}
            \end{minipage}
            \hspace{0.8cm}
            \begin{minipage}{0.35\textwidth}
                \begin{prooftree}
                    \AxiomC{$[L,U] = [\mz, \mz]$}
                    \RightLabel{Empty$^\sharp$}
                    \UnaryInfC{$\dena{\emptyset}([\mathcal{S}]) = [L,U]$}
                \end{prooftree}
            \end{minipage}

            \begin{prooftree}
                \AxiomC{$\forall 0 \leq i,j \leq l+1.~ \big(L_{i,j} = U_{i,j} = (i+1 = j \land s_i = c)\big)$}
                \RightLabel{Character$^\sharp$}
                \UnaryInfC{$\dena{c}([\mathcal{S}]) = [L, U]$}
            \end{prooftree}
            
            \begin{prooftree}    
                \AxiomC{$\dena{r_1}([\mathcal{S}]) = [L_1,U_1] \quad \dena{r_2}([\mathcal{S}]) = [L_2,U_2] \quad [L,U] = [L_1 + L_2,\ U_1 + U_2]$}
                \RightLabel{Union$^\sharp$}
                \UnaryInfC{$\dena{r_1 \mid r_2}([\mathcal{S}])=[L,U]$}
            \end{prooftree}
            
            \begin{prooftree}    
                \AxiomC{$\dena{r_1}([\mathcal{S}]) = [L_1,U_1]$ \quad $\dena{r_2}([\mathcal{S}]) = [L_2,U_2] \quad [L,U] = [L_1 \times L_2,\ U_1 \times U_2]$}\RightLabel{Concat$^\sharp$}
                \UnaryInfC{$\dena{r_1\cdot r_2}([\mathcal{S}])=[L,U]$}
            \end{prooftree}
        
            \begin{prooftree}    
                \AxiomC{$\dena{r}([\mathcal{S}]) = [L_1,U_1]$}
                \AxiomC{$[L,U] = [\mi+\Sigma_{k=1}^{\ell} L_1^k,\ \mi+\Sigma_{k=1}^{\ell} U_1^k]$}
                \RightLabel{Star$^\sharp$}
                \BinaryInfC{$\dena{r^*}([\mathcal{S}])=[L,U]$}
            \end{prooftree}
            \begin{prooftree}
                \AxiomC{$\dena{r}([\mathcal{S}]) = [L_1, U_1]$}
                \AxiomC{$[L, U] = [\mo - U_1, \mo - L_1]$}
                \RightLabel{Neg$^\sharp$}
                \BinaryInfC{$\dena{\neg r}([\mathcal{S}]) = [L,U]$}
            \end{prooftree}
            
        \caption{Abstract semantics}
        \label{fig:regex-sem-abs}
        }
    \end{subfigure}
    \caption{An example-based \semgus problem for regular expressions (\Cref{fig:regex-grammar,fig:regex-sem,fig:regex-constraint}), and an sound abstract semantics for the grammar $G_R$ (\Cref{fig:regex-sem-abs}). We denote $[\mathcal{S}]$ to denote the input interval over strings $s$.}
    \label{fig:regex-example}
\end{figure}

The semantics $\den{\cdot}_R$ has type $M^{(l+1) \times (l+1)} \times M^{(l+1) \times (l+1)} \to M^{(l+1) \times (l+1)}$, which is a map over matrices of size $(l+1) \times (l+1)$, which will encode the set of accepted substrings.
The semantics $\den{\cdot}_S$ has type $M^{(l+1) \times (l+1)} \rightarrow \textit{Boolean}$, a map from a matrix to Boolean type (true/false).
The one production associated with nonterminal $S$, \textit{accepts(R)}, takes in a regular expression represented by nonterminal $R$ and evaluates if that regular expression accepts an input string $s$.
We take this approach to standardize the format, where all examples are constraints that evaluate to true or false.

\subsection{Interval Abstract Semantics}
The semantics defined in \Cref{fig:regex-sem-abs} is a sound interval abstract semantics for $G_R$.
The intervals are defined over matrices, ordered by the relation $A \preceq_M B =_\textit{df} (\bigwedge_{i,j} (A_{ij} \implies B_{ij}))$, as well as the intervals of strings lifted from the input string corresponding to $s$.
The matrix relation captures language inclusion for strings of length $l$: $r_1 \preceq r_2$ iff $\forall s \in \Sigma^l.\ s \in r_1 \implies s \in r_2$.
The partial order over strings is a \textit{substring} order, where $s_1 \preceq_\mathcal{S} s_2$ if and only if $s_2$ has $s_1$ as a prefix, i.e. $s_2 = s_1 + s_i$ for some $s_i$. 
We also have special $\bot, \top$ elements such that $\bot \preceq_\mathcal{S} s \preceq_\mathcal{S} \top$ for any string $s$.

Notice how all of the semantic rules in Figure~\ref{fig:regex-sem} have corresponding abstract semantic rules in Figure~\ref{fig:regex-sem-abs} that are just the rules lifted to intervals.
To preserve the sound-approximation property established in Theorem~\ref{thm:approx}, the abstract semantics of holes $\hole_S$ and $\hole_R$ must cover the set of all states that any derivation from the hole can take on a given input.
Trivially, we do so by setting $\dena{\hole_S}([L,U]) = [\bot, \top]$ and $\dena{\hole_R}([L,U]) = (0\mid1)^*$ (i.e., we assign each hole the widest possible interval abstract semantics), which is what we used in Figure~\ref{fig:regex-sem-abs}.

\begin{example}[Pruning]
    Consider the example partial regular expression from \Cref{sec:overview}.
    We first show how to compute the interval abstraction for $\dena{accepts(0 \cdot \hole_R)}$ on the example $(1, \textit{true})$ using the abstract semantics from Figure~\ref{fig:regex-sem-abs} for $l = 2$.
    For notational simplicity, we write the matrices in our intervals simply as the regular expressions themselves to which these matrices correspond, and include the corresponding matrices below.
    We first start with $\dena{0}(1) = [0, 0]$, an interval over regular expressions.
    We also compute $\dena{\hole_R}(1) = [\emptyset, (0 \mid 1)^*]$.
    Compositionally, we apply these two rules and the associated CHC for the abstract semantics of Concat to get $\dena{0 \cdot \hole_R}(1) = [0 \cdot \emptyset, 0 \cdot (0 \mid 1)^*] = [\emptyset, 0 \cdot (0 \mid 1)^*]$.
    Finally, we apply the abstract semantics of Accepts to get $\dena{accepts(0 \cdot \hole_R)}(1) = [\bot, \bot]$, as
    the matrix for $\emptyset$ is the all-zero matrix $\mz$, and the matrix for $0 \cdot (0 \mid 1)^*$ on input string $1$ has a zero in its upper-right entry.
    Thus, because the required output $\textit{true}$ for positive example $(1, \textit{true})$ is outside of the interval $[\bot, \bot]$ obtained from the abstract semantics, we can prune away the partial program $0 \cdot \hole_R$.
    The abstract semantics, including the $3 \times 3$ upper-triangular matrices computed for a string length of $l = 2$, are explicitly shown below:
    \begin{align*}
        \dena{0}(1) = [\mz, \mz] &= \left[\begin{bmatrix}0 & 0 & 0\\ & 0 & 0 \\ && 0 \end{bmatrix}, \begin{bmatrix}0 & 0 & 0\\ & 0 & 0 \\ && 0 \end{bmatrix} \right] \\
        \dena{\hole_R}(1) = [\mz, \mo] &= \left[\begin{bmatrix}0 & 0 & 0\\ & 0 & 0 \\ && 0 \end{bmatrix}, \begin{bmatrix}1 & 1 & 1\\ & 1 & 1 \\ && 1 \end{bmatrix} \right] \\
        \dena{0 \cdot \hole_R}(1) = [\mz \times \mz, \mz \times \mo] &= \left[\begin{bmatrix}0 & 0 & 0\\ & 0 & 0 \\ && 0 \end{bmatrix}, \begin{bmatrix}0 & 0 & 0\\ & 0 & 0 \\ && 0 \end{bmatrix} \right] \\
        \dena{accepts(0 \cdot \hole_R)}(1) &= [\bot, \bot]
    \end{align*}
\end{example}
  \section{Generalizing Enumeration of Finite Domains}
\label{sec:app:enum-domain}
In~\Cref{sec:nearly-monotonic}, we discussed how to extend our analysis of monotonicity to situations such as if-then-else, where the function may not be totally monotonic but where we can still compute abstract semantics.
We generalize this idea in the following way:

\begin{definition}[Joined Interval]
\label{def:chc:nearlymono}
Suppose that $f(x, y_1, \ldots, y_n)$ is monotonic with respect to all arguments $x, y_1, \ldots, y_n$ except $y_k$.
For this $y_k$, denote the domain that it belongs to as $D_k$.
The \textit{joined interval over $y_k$} is the interval generated by taking the join across all possible instantiations of $y_k \in D_k$:
\begin{equation}
\label{eq:nearlymono}
    \tilde{f}(I_0, \ldots, I_n) = \bigsqcup_{v_k \in I_k \subseteq D_k} \hat{f}_{v_k}(I_0, \ldots, I_{k-1}, I_{k+1}, \ldots I_n)
\end{equation}
where $\hat{f}_{v_k}$ is the endpoint extension from~\Cref{thm:mono-interval-abstraction} applied on the induced function resulting from evaluating $y_k$ concretely as $v_k$. 
\end{definition}

The key is that if $y_k$ belongs to a finite domain $D_{k}$, then the join in~\Cref{eq:nearlymono} is finite and can be easily calculated!
We can now update our abstract semantics defined in~\Cref{def:chc:compilation} by redefining $f^\sharp$.
As before, if $f$ is monotone, let $f^\sharp$ be defined as $\hat{f}$ (as in \Cref{thm:mono-interval-abstraction}).
If $f$ is monotonic with respect to all arguments except $y_{k_1}, \ldots, y_{k_i}$, and $D_{k_1}, \ldots, D_{k_i}$ are finite, then we define $f^\sharp$ as taking the joined interval $\tilde{f}$ (as in \Cref{def:chc:nearlymono}) on each of $y_{k_1}, \ldots, y_{k_i}$.
If neither of the above conditions hold, define it as $f^\sharp(y_1,\ldots,y_n)=\top$.
This new abstract semantics still remains a sound interval abstract semantics, as the following theorem establishes.

\begin{restatable}[Soundness of New Endpoint Interval Semantics]{theorem}{newintervalsoundness}
\label{thm:new-soundness-endpoint-semantics}
Consider redefining the semantics from~\Cref{thm:soundness-endpoint-semantics}, where we update~\Cref{def:chc:compilation} with our new $f^\sharp$ from above.
This new redefined semantics is still a \textit{sound interval abstract semantics}.
\end{restatable}
  \section{Theorems and Proofs}
This section repeats the theorems used in the paper and formally proves their results.

\pruning*
\begin{proof}
    Suppose $\dena{\cdot}$ is a sound interval abstract semantics for $G$, $P$ is a partial program such that $\textsc{Prune}(P, \examples)$ return \textsc{True}, and $P'$ is a complete program derived from $P$, meaning $P \mapsto^* P'$.
    From line~\ref{alg:validprune}, this happens if there is some $(i_k, o_k) \in \examples$ such that $o_k \not \in \dena{P}(i_k)$.
    Also, since $\dena{\cdot}$ is a sound interval abstract semantics, $\den{P'}(x) \in \dena{P}([l,u])$ (Definition~\ref{def:abs-sem}).
    Thus, $\den{P'}(i_k) \in \dena{P}([i_k, i_k])$.
    However, we showed earlier that $o_k \not \in \dena{P}(i_k)$, meaning $\den{P'}(i_k) \neq o_k$.
    Therefore, there doesn't exist any $P'$ derived from $P$ where $P \vdash \examples$.
\end{proof}

\monotonicity*
\begin{proof}
    We first prove soundness. 
    Consider a particular argument $i \in [0..n]$, an interval $[l_i, u_i]$, and element $x_i \in [l_i, u_i]$.
    First, suppose $f$ was increasing in its $i$-th argument, meaning $m_i = \monup$. 
    Fixing all the other arguments, we know that for constants $c_1, \ldots, c_n$, $f(c_1, \ldots, c_{i-1}, l_i, c_{i+1}, \ldots, c_n) \preceq f(c_1, \ldots, c_{i-1}, x_i, c_{i+1}, \ldots, c_n) \preceq f(c_1, \ldots, c_{i-1}, u_i, c_{i+1}, \ldots, c_n)$ (Definition~\ref{def:monotonicity}).
    For notational convenience and for clarity, we write this situation as $f(\ldots, l_i, \ldots) \preceq f(\ldots, x_i, \ldots) \preceq f(\ldots, u_i, \ldots)$, and elide the constants.
    Additionally, since $m_i = \monup$, $a_{m_i}(l_i, u_i) = a_{\monup}(l_i, u_i) = l_i$.
    Thus, $l = f(\ldots, a_\monup(l_i, u_i), \ldots) \preceq f(\ldots, x_i, \ldots)$.
    Similarly, if we suppose $f$ is decreasing in its $i$-th arguments, then $l = f(\ldots, a_{\mondown}(l_i, u_i), \ldots) = f(\ldots, u_i, \ldots) \preceq f(\ldots, x_i, \ldots)$.
    A symmetric argument can be made to show that $f(\ldots, x_i, \ldots) \preceq f(\ldots, b_{m_i}(l_i, u_i), \ldots) = u$.
    Thus, for any $i \in [0..n]$:
    \begin{equation}
        \label{eq:argumentwise-monotonicity}
        f(\ldots, l_i, \ldots) \preceq f(\ldots, x_i, \ldots) \preceq f(\ldots, u_i, \ldots)
    \end{equation}

    We can then proceed to prove soundness as follows:
    Since $f$ is monotone, it is monotonic on every argument $i \in [0..n]$.
    First, we fix some set of constants $c_1 \in Y_1, \cdots, c_n \in Y_n$.
    Applying Equation~\ref{eq:argumentwise-monotonicity} for $i=0$ and fixing the remaining arguments as our constants $c_i$ gives us that $f(l_0, c_1, c_2, \ldots, c_n) \preceq f(x_0, c_1, c_2, \ldots, c_n) \preceq f(u_0, c_1, c_2, \ldots, c_n)$.
    We then iteratively replace constants with our result in Equation~\ref{eq:argumentwise-monotonicity}:
    \begin{gather*}
        f(l_0, c_1, c_2, \ldots, c_n) \preceq f(x_0, c_1, c_2, \ldots, c_n) \preceq f(u_0, c_1, c_2, \ldots, c_n) \\
        f(l_0, l_1, c_2, \ldots, c_n) \preceq f(x_0, x_1, c_2, \ldots, c_n) \preceq f(u_0, u_1, c_2, \ldots, c_n) \\
        \vdots \\
        f(l_0, l_1, l_2, \ldots, l_n) \preceq f(x_0, x_1, x_2, \ldots, x_n) \preceq f(u_0, u_1, u_2, \ldots, u_n)
    \end{gather*}
    The last line proves our soundness result.

    We now prove precision by showing the existence of $x_0^l, \ldots, x_n^l$ and $x_0^u, \ldots, x_n^l$ constructively.
    Consider some argument $i \in [0..n]$. 
    Suppose $f$ is monotonically increasing on its $i$-th argument, meaning $\hat f(\ldots, [l_i, u_i], \ldots) = [f(\ldots, l_i, \ldots), f(\ldots, u_i, \ldots)]$ (where the constants are elided).
    The specific $x_i^l$ and $x_i^u$ are precisely the lower and upper endpoints of the recipe.
    Specifically, let $x_i^l = l_i$ and $x_i^u = l_u$.
    And in the case of $f$ monotonically decreasing on its $i$-th argument, then let $x_i^l = u_i$ and $x_i^u = l_i$.
    If we take this corresponding $x_i^l$ and $x_i^u$ for each of arguments, then $f(x_0^l, \ldots, x_n^l) = l$ and $f(x_0^u, \ldots, x_n^u) = u$.
    
\end{proof}

\intervalsoundness*
\begin{proof}
    From Definition~\ref{def:chc:compilation} and Theorem~\ref{thm:mono-interval-abstraction}, we know that $f^\sharp$ is a sound interval abstraction for $f$.
    This proves that $\dena{\cdot}$ is a sound abstract semantics for $\den{\cdot}$ on the subset $L(G)$ of concrete terms in $L(G_{int})$.
    Definition~\ref{def:hole:abstraction} completes the proof by providing $Hole_N$ a sound hole semantics, as the combination gives a sound interval abstraction to all productions in $G_{int}$ and therefore compositionally to all terms $t \in G_{int}$.
\end{proof}

\preciseholeabstractions*
\begin{proof}
    By~\Cref{def:gfa:eqns}, interval grammar flow analysis is defined as $\denanon{N}{\hole_N}(x) \sqsupseteq \bigsqcup \{ \dena{p}(x) \mid (N \to p) \in G\} \text{ for all } N \in G.$
    By~\Cref{def:abs-sem}, we know that for any partial program $P \in \LL(G_{int})$ that $\forall [l,u], \forall l\preceq x\preceq u, \forall P' \in \LL(G), P \mapsto^* P' \Rightarrow \dennon{S}{P'}(x) \in \denanon{S}{P}([l,u])$.
    Since the semantics $\dena{p}$ of production $N \to p$ is of the (partial) program $P$ with all of its entries as holes, $\forall P' \in \LL(N)$ such that $P \mapsto^* P'$, $\den{P'}_N(x) \in \dena{P}_N([x,x])$ by~\Cref{def:abs-sem}.
    $\dena{P}_N([x,x]) \sqsubseteq \bigsqcup \{ \dena{p}(x) \mid (N \to p) \in G\}$, since $\dena{P}_N([x,x]) = \dena{p}(x)$. 
    Finally, $\bigsqcup \{ \dena{p}(x) \mid (N \to p) \in G\} \sqsubseteq \denanon{N}{\hole_N}(x)$ by~\Cref{def:gfa:eqns}.
    Combining these three results in a chain proves our result.
\end{proof}

\newintervalsoundness*
\begin{proof}
    For simplicity, we consider the case where only a single variable $y_k$ is non-monotonic, as in~\Cref{def:chc:nearlymono}.
    Recall from~\Cref{thm:mono-interval-abstraction} that $\hat{f}$ is a sound interval abstraction for $f$.
    Then we know that each $\hat{f}_{v_k}$ satisfies $\forall x_i \in [l_i, u_i]. f(x_0, \ldots, x_{k-1}, v_k, x_{k+1}, \ldots, x_n) \in \hat{f}_{v_k}([l_0, u_0], \ldots, [l_{k-1}, u_{k-1}], [l_{k+1}, u_{k+1}], \ldots, [l_n, u_n])$ since $f$ was concretely evaluated on $y = v_k$.
    Since each $v_k \in [l_k, u_k]$, then $\forall x_i \in [l_i, u_i]. f(x_0, \ldots, x_n) \in \bigsqcup_{v_k \in [l_k, u_k]} \hat{f}_{v_k}([l_0, u_0], \ldots, [l_{k-1}, u_{k-1}], [l_{k+1}, u_{k+1}], \ldots, [l_n, u_n])$.
    Therefore, $\forall x_i \in [l_i, u_i]. f(x_0, \ldots, x_n) \in \tilde{f}([l_0, u_0, \ldots, [l_n, u_n])$.
    The remainder of the proof exactly follows from the same steps as~\Cref{thm:soundness-endpoint-semantics}.
\end{proof}

\gfatermination*
\begin{proof}
    % \rahul{need to fix this proof now that domain is over intervals, not bounds}
    Suppose the interval bounds $[l_N, u_N]$ belong to a domain $\mathcal{D}$ with no infinite descending chains.
    Assume, by contradiction, that we have an algorithm $\mathcal{A}$ that iteratively computes tighter bounds to some $l_N, u_N$ but performs an infinite number of iterations.
    This means that it will solve for $[l_N, u_N] \sqsubsetneq \dena{\hole_N}_N$.
    Since the semantics is an interval, define $\dena{\hole_N}_N = [l', u']$.
    For $[l', u'] \sqsubset [l_N, u_N]$, this means that $l' \preceq l_N \preceq u_N \preceq u'$ and $l' \prec l_N$ or $u_N \prec u'$ (where the order on domain elements $\prec$ is that induced by the interval order $\sqsubset$).
    At every iteration, it takes the old bounds $l'$ and finds an $l_N$ such that $l' \prec l_N$.
    Since it performs an infinite number of iterations, the sequence of these solved bounds is infinite: $l' \preceq l_1 \preceq l_2 \preceq \cdots$ and $\cdots \preceq u_n \preceq u_2 \preceq u_1 \preceq u'$.
    However, this is a contradiction, as then an infinite descending chain can be created as $[l',u'] \sqsupsetneq [l_1, u_1] \sqsupsetneq [l_1, u_2],\sqsupsetneq \cdots$.
\end{proof}
  \section{Additional Evaluation}

A summary of benchmark results by solver and domain are shown in Table~\ref{tab:summary-by-domain}.

\begin{table}[H]
    \centering
    \begin{tabular}{c|c|ccc|c}
\toprule[.1em]    
        Domain & Total \# & \baseline & \toolnamenogfa & \toolnamegfa & \virtualbest \\
        \midrule[.1em]
        Boolean          & 238  & 118 & 124 & 124 & 124 \\
        Bit Vectors      & 100 & 38 & 44 & 44 & 44 \\
        Matrix Regex     & 39  & 28 & 29 & 29 & 29 \\
        Shallow Regex    & 26  & 18 & 20 & 20 & 20 \\
        CSV              & 17  &  8 & 10 & 15 & 15 \\
        Other Imperative & 10  & 8 & 10 & 10 & 10 \\
        \midrule[.1em]
        Total            & 430 & 217 & 236 & 241 & 241 \\
  \bottomrule[.1em]
\bottomrule[.1em]      
    \end{tabular}    
    \caption{Solved benchmarks for each category}
    \label{tab:summary-by-domain}
\end{table}

  % ... other appendix chapters
}{
  % Code to replace references with some text
  \renewcommand{\refname}{see supplemental material}
}

\end{document}